\documentclass[11pt]{article}
\usepackage{amsmath,amssymb,amsthm} 
\usepackage{mathrsfs}
\usepackage{bbm}
\usepackage{stmaryrd}
\usepackage{paralist}
\usepackage{hyperref}
\usepackage{graphicx}
\usepackage{subcaption}
\graphicspath{{figures/}}
\usepackage{color}
\usepackage{todonotes}
\usepackage{bm}
\usepackage{float}

\setdefaultenum{(i)}{}{}{}
\usepackage[margin=1in]{geometry} 

\numberwithin{equation}{section}

\renewenvironment{thebibliography}[1]{%
\begin{oldthebibliography}{#1}%
\setlength{\baselineskip}{.8em}
\linespread{.8}
\small
\setlength{\parskip}{0ex}%
\setlength{\itemsep}{.1em}%
}%
{%
\end{oldthebibliography}%
}

\theoremstyle{plain}
\newtheorem{theorem}{Theorem}[section]

\newtheorem{lemma}[theorem]{Lemma}

\newtheorem{definition}[theorem]{Definition}
\newtheorem{assumption}[theorem]{Assumption}

\DeclareMathAlphabet\scr{U}{scr}{m}{n}
\SetMathAlphabet\scr{bold}{U}{scr}{b}{n}
  \DeclareFontFamily{U}{scr}{\skewchar\font'177}%
  \DeclareFontShape{U}{scr}{m}{n}{<-6>rsfs5<6-8>rsfs7<8->rsfs10}{}%
  \DeclareFontShape{U}{scr}{b}{n}{<-6>rsfs5<6-8>rsfs7<8->rsfs10}{}%

\theoremstyle{definition}
\newtheorem{remark}[theorem]{Remark}
\newtheorem{example}[theorem]{Example}
\def\E{\mathbb{E}}
\def\R{\mathbb{R}}
\def\sign{\text{sign}}
\def\one{\mathbbm{1}}

\begin{document}

\title{Asset Pricing with General Transaction Costs:\\ Theory and Numerics\footnote{The authors thank Bruno Bouchard, Agostino Capponi, Paolo Guasoni, Robert Pego, and Chen Yang for fruitful discussions, and Steve Shreve for detailed comments on Appendix~\ref{app:ode}. The pertinent remarks of an anonymous associate editor and two anonymous referees are also gratefully acknowledged.}}

\author{Lukas Gonon\thanks{University of Munich, Department of Mathematics, Theresienstra\ss e 39, 80333 Munich, Germany, email \texttt{gonon@math.lmu.de}.}
\and
Johannes Muhle-Karbe\thanks{Imperial College London, Department of Mathematics, London, SW7 1NE, UK, email \texttt{j.muhle-karbe@imperial.ac.uk}. Research supported by the CFM-Imperial Institute of Quantitative Finance.}
\and
Xiaofei Shi\thanks{Carnegie Mellon University, Department of Mathematical Sciences,  5000 Forbes Avenue, Pittsburgh, PA 15213, USA, \texttt{xiaofeis@andrew.cmu.edu}.}
}

\date{April 15, 2020}

\maketitle

\begin{abstract}
We study risk-sharing equilibria with general convex costs on the agents' trading rates. For an infinite-horizon model with linear state dynamics and exogenous volatilities, we prove that the equilibrium returns mean-revert around their frictionless counterparts --  the deviation has Ornstein-Uhlenbeck dynamics for quadratic costs whereas it follows a doubly-reflected Brownian motion if costs are proportional. More general models with arbitrary state dynamics and endogenous volatilities lead to multidimensional systems of nonlinear, fully-coupled forward-backward SDEs. These fall outside the scope of known wellposedness results, but can be solved numerically using the simulation-based deep-learning approach of \cite{han.al.17}. In a calibration to time series of prices and trading volume, realistic liquidity premia are accompanied by a moderate increase in volatility. The effects of different cost specifications are rather similar, justifying the use of quadratic costs as a proxy for other less tractable specifications.

\end{abstract}




\section{Introduction}

The interplay between \emph{liquidity} and \emph{asset prices} has been studied extensively in the empirical literature, cf., e.g., \cite{amihud.al.06} and the references therein for an overview. The analysis of theoretical models consistent with the main stylized facts established in these studies is challenging, however, since both models with limited liquidity and equilibrium asset pricing models are notoriously intractable on their own right. These difficulties are of course only compounded for models where equilibrium asset prices are determined endogenously in the presence of trading frictions. 

Accordingly, tractable models often focus on settings where asset prices~\cite{vayanos.vila.99,lo.al.04,weston.17} or trading volume~\cite{vayanos.98} are deterministic. Models where asset prices and trading volume both fluctuate randomly have recently been analyzed by focusing on quadratic costs on the agents' trading rates~\cite{garleanu.pedersen.16,sannikov.skrzypacz.16,bouchard.al.18,herdegen.al.19}. The analysis of these models crucially exploits the linearity of the corresponding first-order conditions, thereby naturally raising the question how delicately the qualitative and quantitative predictions depend on the specific choice of the trading costs. Typical examples are linear transaction taxes or empirical estimates of actual trading costs that typically correspond to a power of the order flow of around $3/2$~\cite{lillo.al.03,almgren.al.05}.

The present study addresses this challenge by studying risk-sharing equilibria with \emph{general} convex costs levied on the agents' trading rates. This nests quadratic costs as one special case, but also covers proportional costs as another limiting case. We show that in an infinite-horizon model with linear state dynamics and exogenous price volatility, the corresponding equilibrium returns can be characterized \emph{explicitly} up to the solution of a single nonlinear ODE. The latter determines the mean-reverting fluctuations of the frictional equilibrium returns around their frictionless counterparts. If costs are quadratic, this ``liquidity premium'' is an Ornstein-Uhlenbeck process similarly as in \cite{garleanu.pedersen.16,bouchard.al.18,herdegen.al.19}; for proportional costs it turns out to be a doubly-reflected Brownian motion.

To assess the quantitative differences between the respective equilibrium returns, we calibrate our model to market data. This is challenging, since agents' preferences and endowments are not directly observable. However, we show that this difficulty can be overcome as follows. We first pin down some of the parameters by calibrating the frictionless model to a time series of prices. Then, we fit the additional parameters of our model with proportional transaction costs to trading volume data, by exploiting that the average turnover rate in the model can be computed in closed form. To obtain comparable results for other forms of trading costs, we in turn match the corresponding trading volumes and stationary variances of the liquidity premium. 

We find that realistic transaction costs lead to considerable fluctuations around the constant frictionless expected returns if agents' trading targets are calibrated to match the large trading volume observed empirically. In contrast, the differences between the results for proportional, quadratic, and intermediate costs are rather small if the magnitude of these costs is matched appropriately. This provides some justification for the use of quadratic trading costs as a proxy for other less tractable specifications. 

Trading volume is given by a nonlinear function of the equilibrium returns in our model, and this transformation magnifies the differences between different cost specifications. Indeed, for quadratic costs, volume follows the absolute value of an Ornstein-Uhlenbeck process, whereas subquadratic costs skew volume towards either zero or infinite rates as observed in the limiting case of proportional costs. The trading volume dynamics implied by our model recapture the main stylized facts observed empirically, such as autocorrelation and mean reversion~\cite{lo.wang.00}. However, with realistically small transaction costs, our simple stylized model with constant volatilities and trading needs cannot reproduce the strong persistence observed in real time-series data. Likewise, matching the large average turnover rate observed empirically is tied to excessive fluctuations relative to the data. 

Two further key restriction of our tractable benchmark model are that liquidity premia are zero on average and volatilities are given exogenously. The first property is at odds with a large empirical literature that documents that less liquid securities exhibit higher average expected returns~\cite{amihud.mendelson.86,brennan.subrahmanyam.96,pastor.stambaugh.03}. The second rules out studies of the effects changes in market liquidity (e.g., the introduction of a transaction tax) have on market volatility. 

In order to address these limitations, it is natural to extend our baseline model to more general state dynamics (for which volatilities are mean-reverting stochastic processes, for example) and to determine the volatility process endogenously by matching an exogenous terminal dividend. This in turn allows to study how liquidity influences volatility. As a byproduct, models of this kind can also generate systematic liquidity premia as demonstrated in a model with quadratic costs by~\cite{herdegen.al.19}.  However, the analysis of models with endogenous volatilities is substantially more involved, in that it leads to fully-coupled systems of nonlinear forward-backward stochastic differential equations. Indeed, the optimal positions evolve forward from the agents' initial allocations. In contrast, the initial optimal trading rates need to be determined as part of the solution, taking into account that trading stops at the terminal time. Likewise, the stock dynamics also need to be derived from the terminal dividend. For quadratic trading costs, wellposedness of this multidimensional and fully-coupled system has recently been established by~\cite{herdegen.al.19} for agents with sufficiently similar risk aversions. If trading costs are not quadratic, wellposedness of the system is a challenging open problem, and simplifications to systems of coupled Riccati equations as in~\cite{herdegen.al.19} are not possible even for the simplest linear state dynamics.

 In order to nevertheless shed some light on the behaviour of such more general models, we demonstrate in the present study that systems of this kind can be solved numerically by adapting the simulation-based deep learning approach of~\cite{han.al.17} if the time horizon is not too long. Here, the idea is to use a deep neural network to parametrize the ``decoupling field'' that describes the backward components as a function of the forward variables. For each choice of the decoupling field, the corresponding forward dynamics of the system can in turn be simulated by a standard Euler scheme, so that it remains to keep updating the initial guess for the decoupling field using stochastic gradient descent until the simulation matches the terminal condition of the equation sufficiently well. 

We verify that this algorithm produces accurate results by comparing it to the Riccati system that describes the equilibrium in a benchmark example with quadratic costs and linear state dynamics in \cite{herdegen.al.19}. With minor adjustments, the same algorithm is also able to deal with other trading cost specifications. Like in our baseline model, the specification of the trading cost only has a minor effect on the equilibrium price dynamics for our calibrated parameters. Complementing these numerical results with a rigorous verification theorem is an important direction for future research, as is an extension of the model with stochastic volatility that allows to produce richer trading volume dynamics.

The remainder of this article is organized as follows. Section~\ref{s:fl} introduces our frictionless baseline model and derives the corresponding equilibrium returns. In Section~\ref{sec:regular}, this model and the equilibrium results are extended to general smooth convex costs on the agents' trading rates. The limiting case of proportional transaction costs is treated separately in Section~\ref{sec:singular}. Both models are calibrated to time-series data in Section~\ref{sec:calibration}. Equilibrium prices in more general models with arbitrary state dynamics and endogenous volatilities are linked to nonlinear FBSDEs in Section~\ref{sec:general} and solved numerically in~Section~\ref{ssec:numerics}. For better readability, all proofs are collected in Section~\ref{s:proofs} as well as Appendices~\ref{app:ode} and \ref{app:calculation}.

\paragraph{Notation}

Throughout, we fix a filtered probability space $(\Omega,\mathscr{F},(\mathscr{F}_t)_{t \geq 0},\mathbb{P})$ supporting a standard Brownian motion $(W_t)_{t \geq 0}$ and denote by $\mathscr{L}^p$ the adapted processes $(X_t)_{t \geq 0}$ that satisfy $\E[\int_0^T |X_t|^p dt]<\infty$ for all $T>0$.

\section{Frictionless Baseline Model}\label{s:fl}
\subsection{Risk-Sharing Economy}

We consider two agents indexed by $n=1,2$ that receive (cumulative) random endowments 
\begin{align*}
d\zeta^n_t =\beta^n_t dW_t, \quad \mbox{where $\beta^n_t =\beta^n W_t$,} \quad \beta^n \in \mathbb{R}.
\end{align*}
To hedge against the fluctuations of their endowment streams,  the agents trade a safe and a risky asset. The price of the safe asset is exogenous and normalized to one. The price of the risky asset follows
\begin{align*}
dS_t=\mu_t dt+\sigma dW_t.
\end{align*}
Here, the constant volatility $\sigma$ is given exogenously, whereas the expected returns process $\mu \in \mathscr{L}^2$ is to be determined endogenously by matching the agents' demand to the fixed supply $s \in \mathbb{R}$ of the risky asset; see \cite{vayanos.98,zitkovic.12,choi.larsen.15,kardaras.al.15,garleanu.pedersen.16,xinZit17,bouchard.al.18} for related equilibrium models where the volatility is a free parameter. Models where the volatility is determined endogenously are discussed in Section~\ref{sec:general}.

\begin{remark}
Unlike for more general preferences, an additional orthogonal component (and a finite variation drift) of the agents' endowments would not change the optimizers of the simple linear-quadratic goal functionals~\eqref{eq:nocosts}, \eqref{eq:goal2} that we consider below, compare~\cite{bouchard.al.18}. We therefore focus on the present most parsimonious specification.

The restriction to two agents is made to reduce the dimensionality of the problem. More agents can be treated without difficulties in the frictionless case and, using matrix algebra, also for quadratic costs~\cite{bouchard.al.18}. For more general transaction costs, however, more than two agents would lead to multidimensional nonlinear differential equations. Therefore, we focus on two (representative) here agents for tractability.

Likewise, we restrict ourselves to an extremely specific endowment volatility in order to avoid introducing additional state variables for the optimization problems with transaction costs.

Finally, the assumption of a constant exogenous volatility is also crucial for obtaining analytical results for general transaction costs in Section~\ref{sec:regular} and \ref{sec:singular} below. In contrast, models where the volatility is determined endogenously are much more difficult to analyze, as discussed in Section~\ref{sec:general}. Numerical results for such models are reported in Section~\ref{ssec:numerics}.
\end{remark}

\subsection{Frictionless Optimization and Equilibrium}\label{sec:frictionless}

As a reference point, we first consider the frictionless version of the model. Starting from fixed initial positions that clear the market, $\varphi^1_{0-}+\varphi^2_{0-}=s$, the agents choose their positions $\varphi \in \mathscr{L}^2$ in the risky asset to maximize one-period expected returns penalized for the corresponding variances.
Without transaction costs, the continuous-time version of this criterion is
\begin{align}
\bar{J}_T^n(\varphi) &=\E\left[\int_0^T (\varphi_t dS_t+d\zeta^n_t)-\frac{\gamma^n}{2}d\langle \textstyle{\int_0^\cdot} \varphi_u dS_u+\zeta^n\rangle_t \right] =\E\left[\int_0^T \Big(\varphi_t \mu_t -\frac{\gamma^n}{2}(\sigma\varphi_t+\beta_t^n)^2\Big)dt\right]. \label{eq:nocosts}
\end{align}
Put differently, agents trade off expected returns against the tracking error relative to the exogenous target position $-\beta^n/\sigma$ as in~\cite{choi.al.18,sannikov.skrzypacz.16}. The optimal strategy for~\eqref{eq:nocosts} is readily determined by pointwise optimization as 
\begin{align*}
\varphi^n_t = \frac{\mu_t}{\gamma^n \sigma^2} -\frac{\beta_t^n}{\sigma}, \quad t \in [0,T].
\end{align*}
The equilibrium return is in turn pinned down by matching the agents' total demand $\varphi^1_t+\varphi^2_t$ to the supply $s$ of the risky asset at all times $t \in [0,T]$:

\begin{equation}\label{eq:eqnocosts}
\bar\mu_t= \bar{\gamma}\left[s\sigma^2+\sigma(\beta_t^1+\beta_t^2)\right],   \quad t \in [0,T], \quad \mbox{where } \bar{\gamma}=\frac{\gamma^1\gamma^2}{\gamma^1+\gamma^2}.
\end{equation}
The agents' optimal trading strategies corresponding to this frictionless equilibrium return are

\begin{align*}
\bar\varphi^1_t =\frac{s\gamma^2}{\gamma^1+\gamma^2} + \frac{\gamma^2\beta^2_t-\gamma^1\beta^1_t}{(\gamma^1+\gamma^2)\sigma} 
\qquad \bar\varphi^2_t =s-\bar{\varphi}^1_t, \qquad t \in [0,T].
\end{align*}

Note that the frictionless equilibrium return and the corresponding optimal trading strategies are independent of the time horizon $T$. In particular, the frictionless optimizers also maximize the long-run average performance $\bar{J}^n_T/T$ as $T \to \infty$, in that
\begin{align*}
\limsup_{T \to \infty} \frac{1}{T}\left[\bar{J}_T^n(\varphi)-\bar{J}_T^n(\bar\varphi^n)\right] \leq 0, \qquad \mbox{for all competing admissible strategies $\varphi$.}
\end{align*}
With transaction costs -- where the optimizers are no longer independent of the planning horizon -- we will directly solve the long-run version of~\eqref{eq:nocosts}, see Definitions~\ref{def:eq} and \ref{def:eq2} below.

\section{Equilibrium with Costs on the Trading Rate}\label{sec:regular}

\subsection{Costs and Strategies}

We now take into account transaction costs. A popular class of models originating from the optimal execution literature \cite{almgren.chriss.01,almgren.03} focuses on absolutely continuous trading strategies,
\begin{align*}
\varphi_t=\varphi^n_{0-}+\int_0^t \dot{\varphi}_u du, \quad t \geq 0,
\end{align*}
and penalizes the trading rate $\dot{\varphi}_t=d\varphi_t/dt$ with an \emph{instantaneous trading cost} $G(\dot{\varphi}_t)$. Portfolio choice problems for the most tractable quadratic specification $G(x)=\lambda x^2/2$, $\lambda>0$ are analyzed in single-agent models by \cite{garleanu.pedersen.16,almgren.li.16,moreau.al.17,guasoni.weber.17}; equilibrium returns are determined in \cite{garleanu.pedersen.16,sannikov.skrzypacz.16,bouchard.al.18}. In~\cite{guasoni.weber.18,caye.al.19,bayraktar.al.18}, single-agent models are solved for the more general power costs $G(x)=\lambda |x|^q/q$, $q \in (1,2]$ proposed by \cite{almgren.03}. Below, we will determine equilibrium returns for general smooth convex cost functions $G$ as studied in the duality theory of \cite{guasoni.rasonyi.15}:

\begin{assumption}\label{cond:cost}
\begin{enumerate}
\item The trading cost $G: \mathbb{R} \to \mathbb{R}_+$ is convex, symmetric, and strictly increasing on $[0,\infty)$, differentiable on $[0,\infty)$, and satisfies $G(0)=0$;
\item The derivative $G'$ is  also strictly increasing and differentiable on $(0,\infty)$ with $G'(0)=0$;
\item There exist constants $C>0$, $k\geq2$ and $x_0>0$ such that 
\begin{align*}
|(G')^{-1}(x)| \leq C(1+|x|^{k-1})\ \mbox{for all $x\in\R$}, \quad \quad 
G''(x) \leq C \ \mbox{for all $|x|>x_0$}.
\end{align*}
\end{enumerate}
\end{assumption}

One readily verifies that the power functions $G(x)=\lambda |x|^q/q$, $q \in (1,2]$ proposed in \cite{almgren.03} satisfy all of these requirements, as do linear combinations of these power functions. A relevant example beyond the power class is provided by the empirical estimates of~\cite{bucci.al.19}, who find that trading costs generated by price impact are quadratic for small trades but scale with a power of approximately 3/2 for larger order sizes.

With transaction costs, the analogue of the frictionless mean-variance goal functional~\eqref{eq:nocosts} is
\begin{equation}\label{eq:goal2}
J^n_T(\dot{\varphi}) = \E\left[\int_0^T \Big(\varphi_t\mu_t -\frac{\gamma^n}{2}(\sigma\varphi_t+\beta^n_t)^2-G(\dot{\varphi}_t)\Big)dt\right]. 
\end{equation}
Unlike its frictionless counterpart, this optimization problem is no longer ``myopic'', since the current position influences future choices in the presence of transaction costs, and since optimal strategies naturally depend on a finite time horizon $T$ here. To simplify the analysis below, we therefore focus on the ergodic limit of~\eqref{eq:goal2}, where the goal is to maximize the long-run average performance $J^n_T(\dot{\varphi})/T$ as $T \to \infty$. This criterion has a long history in single-agent problems with transaction costs, cf.~\cite{dumas.luciano.91,taksar.al.88,bouchaud.al.12,gerhold.al.14,guasoni.weber.17}. Here, we show that it also makes the equilibrium analysis of general trading costs tractable. Throughout, we focus on \emph{admissible} strategies 
\begin{align*}
\varphi_t=\varphi^n_{0-}+\int_0^t \dot{\varphi}_u du, \quad t \geq 0
\end{align*}
that satisfy the integrability conditions
\begin{align}\label{eq:integrable}
\E\left[ \int_0^T G(\dot{\varphi}_t)dt \right] < \infty, \qquad \E\left[ \int_0^T \varphi_t^{2}dt \right] < \infty, \quad  \mbox{for all $T >0$},
\end{align}
as well as the transversality condition
\begin{align}\label{assump:vanishing}
\lim_{T\to \infty} \frac {1} {T^2} \E[\varphi_T^2] = 0.
\end{align}
For infinite-horizon goal functionals that don't restrict wealth processes to be positive, similar conditions ruling out arbitrarily large risky positions are also imposed in~\cite{lo.al.04,guasoni.muhlekarbe.15,ekren.muhlekarbe.19}, for example. 

\subsection{Equilibrium}

\begin{definition}\label{def:eq}
$\mu \in \mathscr{L}^2$ is a \emph{(long-run) equilibrium return} if there exist admissible trading rates $\dot\varphi^1$, $\dot\varphi^2$ for agents 1 and 2 such that:
\begin{enumerate}
\item[(Market Clearing)] The total demand $\varphi^1+\varphi^2$ matches the supply $s$ of the risky asset at all times;
\item[(Individual Optimality)] The trading rate $\dot{\varphi}^n$ is optimal for the long-run version of agent $n$'s control problem~\eqref{eq:goal2} in that,
\begin{equation}
\limsup_{T\to \infty} \frac{1}{T} \left[ J_T^n(\dot{\varphi}) - J_T^n(\dot{\varphi}^n)\right] \leq 0, \quad \mbox{{\rm for all competing admissible trading rates $\dot\varphi$.}}
\end{equation}
\end{enumerate}
\end{definition}

\begin{remark}\label{eq:deadweight}
It is important to note that as in, e.g., \cite{lo.al.04,buss.dumas.17}, our transaction cost is an exogenous deadweight cost and not an output of the trading process in equilibrium.
\end{remark}

The construction of the equilibrium return is based on the solution of a nonlinear ODE. For single-agent models with instantaneous trading costs of power form, a corresponding equation has been introduced and studied by \cite{guasoni.weber.18}.\footnote{Indeed, if $G(x)=\lambda |x|^q/q$, $q \in (1,2]$, then differentiating the first-order ODE~(15) in~\cite[Theorem~4.1]{guasoni.weber.18} and a change of variables as in Appendix~\ref{app:calculation} lead to the second-order ODE~\eqref{eqn:ergodic ODE}. An analogous link to a first-order equation is exploited in our existence proof in Appendix~\ref{app:ode}.\label{ODE relationship}}  In Appendix~\ref{app:ode}, we show that their existence and uniqueness proof can be extended to general cost functions satisfying Assumption~\ref{cond:cost}.

\begin{lemma}\label{ODE}
Suppose the instantaneous trading cost $G$ satisfies Assumption~\ref{cond:cost}. Then the ordinary differential equation
\begin{equation}\label{eqn:ergodic ODE}
\frac 1 2 \left(\frac{\gamma^1\beta^1 - \gamma^2\beta^2}{(\gamma^1+\gamma^2)\sigma}\right)^2 g''(x) + g'(x)(G')^{-1}(g(x)) = \frac{(\gamma^1+\gamma^2)\sigma^2}{2} x
\end{equation}
has a unique solution $g$ on $\R$ such that $xg(x)\leq 0$ for all $x\in\R$. Moreover, $g$ satisfies 
\begin{equation}\label{conditions}
\lim_{x\to -\infty} \frac{g(x)}{(G^*)^{-1}(\frac{(\gamma_1+\gamma_2)\sigma^2}{4}x^2)} = 1, \quad \quad
\lim_{x\to +\infty} \frac{g(x)}{(G^*)^{-1}(\frac{(\gamma_1+\gamma_2)\sigma^2}{4}x^2)} = -1,
\end{equation}
where $G^*$ is the Legendre transform of $G$. 
\end{lemma}

\begin{proof}
See Appendix~\ref{app:ode}.
\end{proof}

With the function $g$ from Lemma~\ref{ODE}, we can now define the ergodic state variable that will drive both the expected returns and optimal trading rates in equilibrium.

\begin{lemma}\label{lem:sde}
Let $g$ be the solution of the ODE~\eqref{eqn:ergodic ODE} from Lemma~\ref{ODE}. There exists a unique strong solution of the SDE
\begin{equation}\label{eqn:ergodic SDE}
dX_t = (G')^{-1}\left(g(X_t)\right)dt +  \frac{\gamma^1\beta^1 - \gamma^2\beta^2}{(\gamma^1+\gamma^2)\sigma} dW_t, \quad  t \geq 0 , \quad X_0 = \varphi^1_{0-} - \frac{s\gamma^2}{\gamma^1+\gamma^2}.
\end{equation}
This process is a recurrent diffusion.
 \end{lemma}
 
 \begin{proof}
See Section~\ref{sec:proofsregular}.
\end{proof}

\begin{remark}\label{OU process}
If the instantaneous trading cost is quadratic, $G(x) = \lambda x^2/2$, then $(G')^{-1}(x) =  x/\lambda$, and the solution of~\eqref{eqn:ergodic ODE} from Lemma~\ref{ODE} is $g(x) = -\sqrt{\frac{(\gamma^1+\gamma^2)\sigma^2\lambda}{2}}x$. Accordingly, the dynamics~\eqref{eqn:ergodic SDE} simplify to
 \begin{align*}
dX_t = - \sqrt{\frac{(\gamma^1+\gamma^2)\sigma^2}{2\lambda}} X_tdt +  \frac{\gamma^1\beta^1 - \gamma^2\beta^2}{(\gamma^1+\gamma^2)\sigma} dW_t,  \quad  t \geq 0 , \quad X_0 = \varphi^1_{0-} - \frac{s\gamma^2}{\gamma^1+\gamma^2}.
\end{align*}
Whence, $X$ is an Ornstein-Uhlenbeck process in this case. In general, the drift rate in~\eqref{eqn:ergodic SDE} describes the nonlinear attraction of the process $X$  towards its average level zero, where $xg(x) \leq 0$ ensures that the process is indeed mean reverting and in turn converges to an ergodic limit.
\end{remark}

We now present our first main result. It identifies the equilibrium return for general smooth, convex cost functions.

\begin{theorem}\label{thm:regular}
Recall $\bar{\gamma}=\frac{\gamma^1\gamma^2}{\gamma^1+\gamma^2}$. With the solution $(X_t)_{t \geq 0}$ of~\eqref{eqn:ergodic SDE}, define
\begin{equation}\label{eqn:return dynamic1}
\mu_t =\bar{\gamma}\left[s\sigma^2+\sigma(\beta_t^1+\beta_t^2)\right]+\frac{(\gamma^1 - \gamma^2)\sigma^2}{2} X_t, \quad t \geq 0.
\end{equation}
Then, the trading rates
\begin{equation}\label{strategy}
\dot{\varphi}^1_t = (G')^{-1}\left(g(X_t)\right), \quad \dot{\varphi}^2_t = -(G')^{-1}\left(g(X_t)\right), \quad t \geq 0
\end{equation}
clear the corresponding market and are individually optimal in the long run. Therefore, $(\mu_t)_{t \geq 0}$ is an equilibrium return. 
\end{theorem}

 \begin{proof}
See Section~\ref{sec:proofsregular}.
\end{proof}

The first term in~\eqref{eqn:return dynamic1} is the frictionless equilibrium return from~\eqref{eq:eqnocosts}. Accordingly, the second term describes how the equilibrium return changes due to transaction costs. Evidently, if both agents have the same risk aversion, then the adjustment is zero like for the quadratic costs studied by~\cite{bouchard.al.18}. In this case, both agents are adversely affected by the transaction costs, but the market still clears at the frictionless equilibrium price. 

For heterogenous agents, there is a nontrivial liquidity premium depending on the current demand imbalance. Indeed, in equilibrium, the state dynamics $dX_t$ also describe the evolution of the deviation between agent 1's actual position and its frictionless counterpart,
\begin{align*}
dX_t = (G')^{-1}\left(g(X_t)\right)dt +\frac{\gamma^1\beta^1-\gamma^2\beta^2}{(\gamma^1+\gamma^2)\sigma}dW_t = d(\varphi^1_t-\bar{\varphi}^1_t).
\end{align*}
By market clearing, the sign is reversed for agent 2.  Accordingly, the liquidity premium is positive if the more risk averse agent wants to sell and negative if the more risk averse agent wants to buy to move closer to the corresponding frictionless allocation. In each case, the return adjustment ensures market clearing by offsetting the more risk averse agent's stronger motive to trade.

For quadratic costs, we recover the Ornstein-Uhlenbeck returns from \cite[Corollary~5.5]{bouchard.al.18}. For general convex trading costs, these are replaced by processes with nonlinear mean-reversion speeds. 

\section{Equilibrium with Proportional Costs}\label{sec:singular}

One important cost specification is not covered by Assumption~\ref{cond:cost}: proportional transaction costs. These arise as the limit $p \to 1$ in the model of \cite{almgren.03}. Rather than studying the (singular) limiting behaviour of the corresponding optimal strategies as in \cite{guasoni.weber.18}, we instead show that the equilibrium with proportional costs can be constructed directly using singular rather than regular stochastic control.

Since proportional costs only penalize trade size but not speed, risky positions are naturally described by general finite-variation processes in this case or, equivalently, by their Jordan-Hahn decompositions into minimal increasing processes -- the cumulative numbers of shares purchased and sold:
$$
\varphi_t = \varphi^n_{0-} +\varphi^{\uparrow}_t - \varphi^{\downarrow}_t. 
$$
As in  \cite{janecek.shreve.10,martin.schoeneborn.11,bouchaud.al.12,martin.14} we assume for simplicity that the (cumulative) costs $\lambda(\varphi^\uparrow_T+\varphi^\downarrow_T)$, $\lambda>0$, are proportional to the number of shares traded (rather than the monetary amount transacted). Agent $n$'s goal functional in turn becomes
\begin{equation}\label{eq:goal3}
J^n_T(\varphi) = \E\left[\int_0^T \left(\varphi_t\mu_t -\frac{\gamma^n}{2}(\sigma\varphi_t+\beta^n_t)^2\right)dt -\lambda (\varphi^{\uparrow}_T + \varphi^{\downarrow}_T) \right].
\end{equation}
We again focus on the long-run average performance $J^n_T(\varphi)/T$ as $T \to \infty$ of \emph{admissible strategies} that satisfy the integrability condition 
\begin{align}\label{assump:singular}
\E\left[\int_0^T \varphi_t^2 dt\right]<\infty, \qquad\E[\varphi^\uparrow_T+\varphi^\downarrow_T] <\infty, \quad \mbox{for all $T>0$,}
\end{align}
as well as the transversality condition
\begin{align}\label{singular:vanishing}
\lim_{T\to \infty} \frac {1} {T} \E\left[ |\varphi_T| \right] = 0.
\end{align}

\subsection{Equilibrium}

We use an analogous notion of Radner equilibrium as in Definition~\ref{def:eq}:

\begin{definition}\label{def:eq2}
$\mu \in \mathscr{L}^2$ is a \emph{(long-run) equilibrium return} if there exist admissible strategies $\varphi^1$, $\varphi^2$ for agents 1 and 2 such that:
\begin{enumerate}
\item[(Market Clearing)] The total demand $\varphi^1+\varphi^2$ matches the supply $s$ of the risky asset at all times;
\item[(Individual Optimality)] The strategy $\varphi^n$ is optimal for the long-run version of agent $n$'s control problem~\eqref{eq:goal3} in that,
\begin{equation}
\limsup_{T\to \infty} \frac{1}{T} \left[ J_T^n(\varphi) - J_T^n(\varphi^n)\right] \leq 0, \quad \mbox{for all competing admissible strategies $\varphi$.}  
\end{equation}
\end{enumerate}
\end{definition}

The construction of the equilibrium return with proportional costs is based on the analogue of the mean-reverting process from Lemma~\ref{lem:sde}. This turns out to be a doubly-reflected Brownian motion,
\begin{align}\label{reflected BM}
dX_t = \frac{\gamma^1\beta^1 - \gamma^2\beta^2}{(\gamma^1+\gamma^2)\sigma} dW_t + dL_t - dU_t, 
\end{align}
where $X_{0-} =\varphi_{0-}^1 -\frac{s\gamma^2}{\gamma^1+\gamma^2}$ and $L$, $U$ are the minimal increasing processes with $L_{0-}=U_{0-}=0$ that keep $(X_t)_{t \geq 0}$ in the interval $[-l, l]$,\footnote{See \cite{kruk2007explicit} for the pathwise construction of $L$, $U$. In particular, there is an initial jump in $L$ or $U$ if the initial value $X_{0-}$ lies below $-l$ or above $l$, respectively.
On $(0,T]$, $L$ and $U$ have continuous paths.} 
whose endpoints have the following explicit expression:
\begin{equation}\label{eq:ell}
l = \left({\frac{3\lambda(\gamma^1\beta^1-\gamma^2\beta^2)^2}{(\gamma^1+\gamma^2)^3\sigma^4}}\right)^{\frac 1 3}.
\end{equation}

With the state variable $X$ at hand, we can now formulate our second main result. It shows that the equilibrium return with proportional costs can be expressed in direct analogy to its counterpart for the smooth, superlinear costs treated in Theorem~\ref{thm:regular}. The only difference is that the mean-reverting state variable in Theorem~\ref{thm:regular} is replaced by the doubly-reflected Brownian motion from~\eqref{reflected BM}.

\begin{theorem}\label{thm:singular}
Recall $\bar{\gamma}=\frac{\gamma^1\gamma^2}{\gamma^1+\gamma^2}$. With the solution $(X_t)_{t \geq 0}$ of~\eqref{reflected BM}, define
\begin{equation}\label{eqn:return dynamic}
\mu_t =\bar{\gamma}\left[s\sigma^2+\sigma(\beta_t^1+\beta_t^2)\right]+\frac{(\gamma^1 - \gamma^2)\sigma^2}{2} X_t, \quad t \geq 0.
\end{equation}
Then, the trading strategies
\begin{equation}\label{singular strategy}
\varphi^1_t = \varphi^{1}_{0-}+L_t -U_t, \quad \varphi^2_t = \varphi^2_{0-}+U_t -L_t, \quad t \geq 0,
\end{equation}
clear the market and are individually optimal in the long run. Therefore $(\mu_t)_{t \geq 0}$ is an equilibrium return.
\end{theorem}

 \begin{proof}
See Section~\ref{sec:proofssingular}.
\end{proof}

Note that, in equilibrium, each agent's singular control problem has a fully explicit solution. Similar closed-form expressions for optimal no-trade regions also obtain for the ergodic control of Brownian motion, which underlies the tractability of problems with \emph{small} transaction costs~\cite{soner.touzi.13,kallsen.muhlekarbe.17,cai.al.17}. Surprisingly, the equilibrium constructed in Theorem~\ref{thm:singular} displays the same tractability, even though the corresponding equilibrium return is not zero but a reflected Brownian motion. 

\section{Calibration} \label{sec:calibration}

To assess the quantitative properties of our equilibrium returns, we now calibrate the model to price and trading-volume data for the US equity market. More specifically, we consider the 320 current constituents of the S\&P500 for which ten years of uninterrupted data are available from January~2, 2009 to January~2, 2019 in the CRSP database. (We do not work with an even longer time series, since the corresponding larger changes in price levels then become problematic for our arithmetic model.) To obtain the price dynamics of a ``typical stock'', we then compute the capitalization-weighted average of the respective prices. The total number of outstanding shares of this average stock then is the number of shares outstanding for all our stocks. Likewise, the total share turnover is also aggregated across all stocks. The corresponding time series are available in the online appendix of the present paper.

\subsection{Calibration of the Frictionless Baseline Model}  \label{subsec:calibrationEx}

We first consider the frictionless baseline version of the model from Section~\ref{sec:frictionless}. The exogenous (absolute) daily volatility $\sigma$ can be estimated directly from the time series of stock prices, leading to $\sigma=1.88$ for our dataset.\footnote{Since our average stock prices are 124.11, this corresponds to a Black-Scholes volatility of around $23.8\%$.} To obtain a simple parsimonious model for the equilibrium returns, we suppose throughout as in \cite{lo.al.04} that there is no aggregate endowment ($\beta^1_t=-\beta^2_t$). Then, the frictionless equilibrium expected return from~\eqref{eq:eqnocosts} is $\bar{\mu}=\bar{\gamma}s\sigma^2$. As the number of shares outstanding is $s = 2.46\times 10^{11}$, we choose $\bar{\gamma}=8.31\times 10^{-14}$ to match this to the average (absolute) daily returns $\bar{\mu}$ of $0.072$ in our time series.\footnote{This corresponds to a yearly Black-Scholes return of $14.44\%$ relative to the average price level.}

\subsection{Calibration with Transaction Costs}\label{ss:calcosts}

Whereas the frictionless equilibrium price only depends on the aggregate risk aversion $\bar{\gamma}=\frac{\gamma^1\gamma^2}{\gamma^1+\gamma^2}$ and aggregate endowment $\beta^1+\beta^2$, the individual values of these parameters need to be pinned down to determine equilibria with transaction costs. Moreover, the initial allocations fo the agents need to be specified and an appropriate estimate for the respective trading cost is evidently needed. 

\paragraph{Proportional Costs}
For proportional costs, we use the estimate obtained in \cite{novy.velikov.16} for value-weighted trading strategies: $0.25\%$ of the average stock prices, that is $\lambda_1=0.31$ for our dataset.\footnote{Somewhat larger  bid-ask spreads of 1\% are used \cite{lynch.tan.11,buss.dumas.17}, for example.} Once the aggregate risk aversion $\bar{\gamma}$ is fixed, the individual agents' absolute risk aversions $\gamma^1$, $\gamma^2$ are free parameters in the present model, which correspond to the agents' sizes relative to each other. If both agents are of the same size, the frictional equilibrium coincides with its frictionless counterpart.  To illustrate the effect of heterogeneity, we set $\gamma^2=2\gamma^1$, so that agent 2 has half the risk capacity of agent 1. (These specific parameter values are chosen because they lead to realistic levels of liquidity premia in the extended version of the model with endogenous volatilities, see Section~\ref{subsec:numericalResults}.) Then, with $\bar{\gamma}=8.31\times 10^{-14}$ we have $\gamma^1=1.25\times 10^{-13}$ and $\gamma^2=2.5\times10^{-13}$. For the initial allocations, we suppose for simplicity that each agent initially holds a fraction of the total supply equal to their share of the total risk tolerance, $\varphi^1_{0-}=\frac{\gamma^2}{\gamma^1+\gamma^2}s=s-\varphi^2_{0-}$. This minimizes the effect of transaction costs because no initial bulk trades are necessary in this case. But the initial allocation generally only affects the initial conditions of the state variables in our long-run equilibria in Theorems~\ref{thm:singular}, so that the effect of different specifications is quickly averaged out in any case.

Finally, we calibrate the value of the endowment volatilities $\beta_1^1=-\beta_1^2=\beta_1$ to time-series data for trading volume. More specifically, given our estimate $\lambda_1=0.312$ from the proportional cost, we choose the parameter $\beta_1$ to match the average daily share turnover in 2009-2018, which is $\text{ShTu}=1.84 \times 10^9$ (that is, about $187\%$ of the outstanding shares per year), to the corresponding long-term average value in our model. Using the ergodic theorem, the latter can be calculated as in~\cite[Lemma~C.2.]{gerhold.al.14}, 
$$
\text{ShTu} =\lim_{T \to \infty} \frac1T \int_0^Td|\varphi^1|_t  = \lim_{T \to \infty} \frac{L_T}{T}+  \lim_{T \to \infty} \frac{U_T}{T} = \frac{1}{2l}\left(\frac{\gamma^1\beta^1_1 - \gamma^2\beta^2_1}{(\gamma^1+\gamma^2)\sigma}\right)^2
=\left(\frac{\gamma^1+\gamma^2}{24\lambda_1 \sigma^2}\right)^{1/3}\beta_1^{4/3} \quad {\rm a.s.}
$$
Accordingly, we choose
$$
\beta_1= \left(\frac{24\text{ShTu}^3 \lambda_{1}\sigma^2}{\gamma^1+\gamma^2}\right)^{1/4} 
= 2.57 \times 10^{10}.
$$

\paragraph{Superlinear Costs}
For comparison, we also consider the power costs $G_q(x)=\lambda_q |x|^q/q$, $q \in (1,2]$. To choose the endowment volatilities $\beta^1_q=-\beta^2_q=\beta_q$ in this case, apply the ergodic theorem to compute the long-term average of the daily share turnover as
$$
\text{ShTu}= \lim_{T \to \infty}\frac1T \int_0^T |\dot{\varphi}^1_t|dt = \int_{-\infty}^\infty \left|(G_q')^{-1}\left(g_q(x)\right)\right| \nu_q(x) dx \quad {\rm a.s.}
$$
Here, $\nu_q(x)$ is the invariant density of the stationary law of the state variable $X$. For quadratic costs $G_2(x)=\lambda_{2} x^2/2$, this is an Ornstein-Uhlenbeck process (cf.~Remark~\ref{OU process}) whose stationary distribution is Gaussian with mean zero and variance $\sqrt{\lambda_2\beta_2^4}/\sigma^3\sqrt{2(\gamma^1+\gamma^2)}$. As $(G_2')^{-1}\left(g_q(x)\right)=-\sqrt{(\gamma^1+\gamma^2)\sigma^2/2\lambda_2}x$ for quadratic costs, the average turnover per year in turn is proportional to the endowment volatility $\beta_2$ in this case,
$$
\lim_{T\to \infty}\frac1T \int_0^T |\dot{\varphi}^1_t|dt = \sqrt{\frac{(\gamma^1+\gamma^2)\sigma^2}{2\lambda_2}} \sqrt{\frac{2}{\pi}} \frac{\beta_2 \lambda_2^{1/4}}{(2(\gamma^1+\gamma^2)\sigma^6)^{1/4}}=\left(\frac{\gamma^1+\gamma^2}{2\pi^2\sigma^2\lambda_2}\right)^{1/4}\beta_2 \quad {\rm a.s.}
$$
Accordingly, to match the average share turnover for a given quadratic transaction cost $\lambda_2$, we need $\beta_2 =\text{ShTu}/(\frac{\gamma^1+\gamma^2}{2\pi^2\sigma^2\lambda_2})^{1/4}$. Whence, it remains to choose an appropriate value for the trading cost parameter $\lambda_2$. To make its impact comparable to the proportional cost, we choose it to obtain the same stationary variance of the state variable $X$ as with proportional costs. 

\begin{figure}[htbp] 
    \centering
        \includegraphics[width=0.40\textwidth]{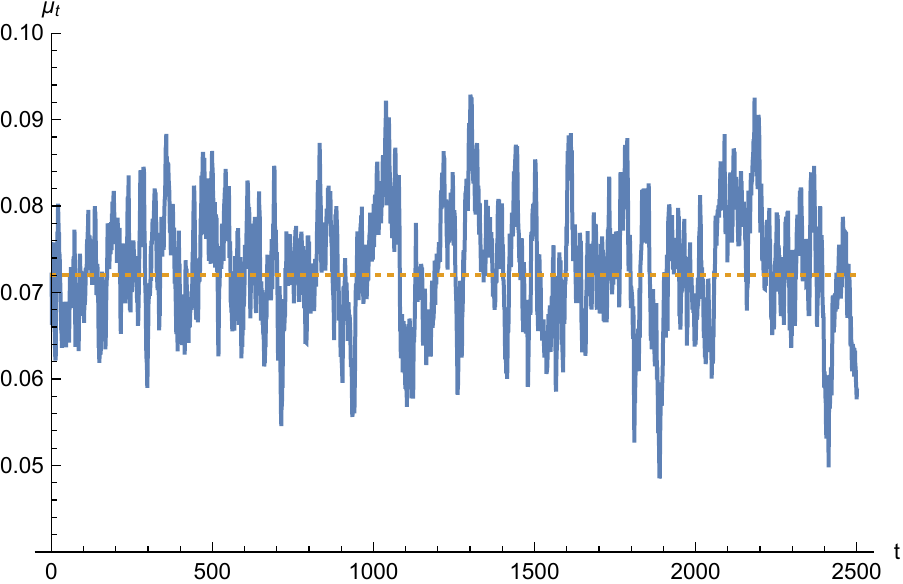}
         \includegraphics[width=0.40\textwidth]{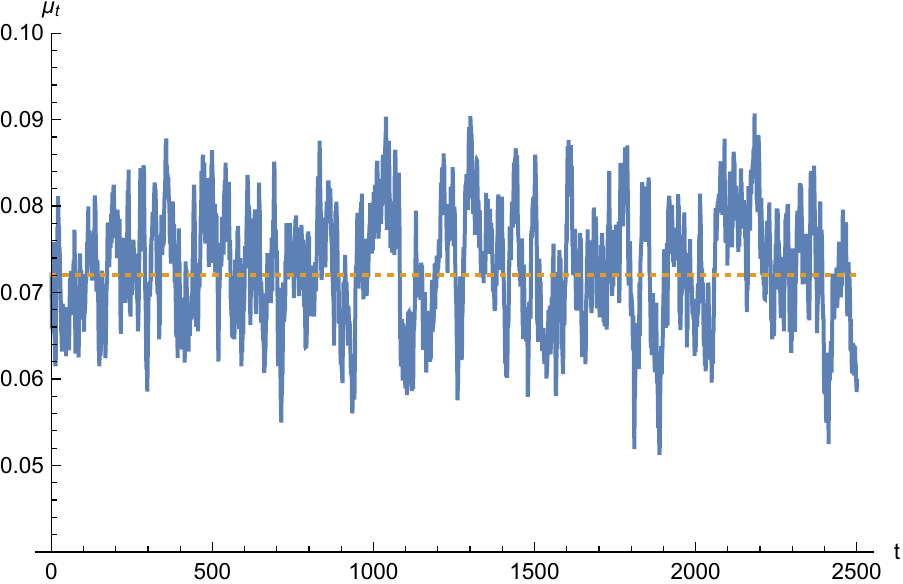}
        \includegraphics[width=0.40\textwidth]{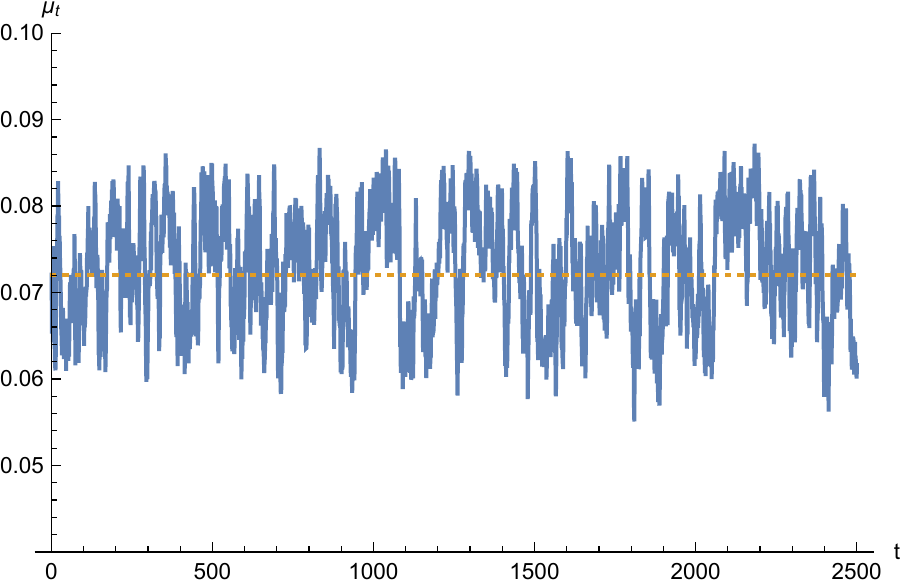}
        \includegraphics[width=0.40\textwidth]{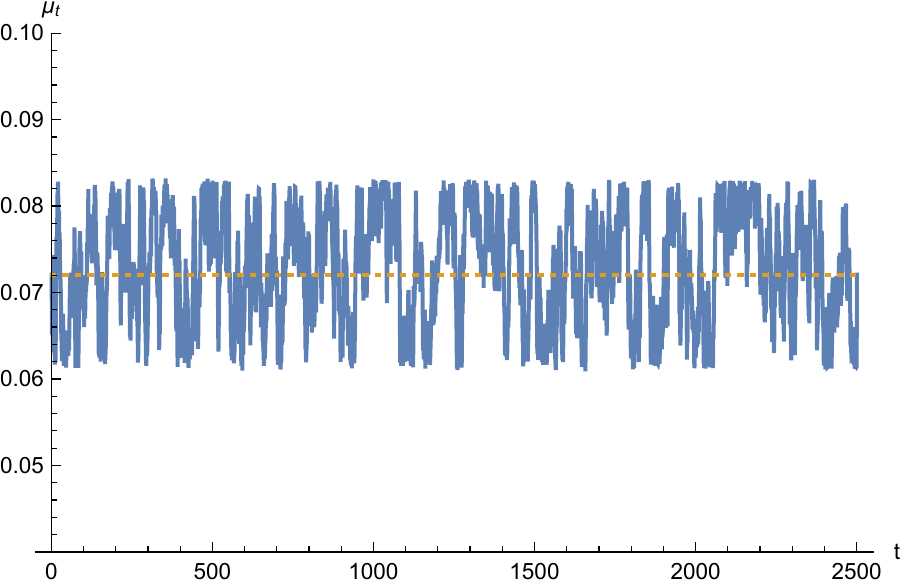}
    \caption{Simulated frictional equilibrium returns with calibrated parameters for quadratic trading costs (left upper panel), costs proportional to the $3/2$-th power of the agents' trading rates (right upper panel), to the $9/8$-th power (lower left panel), and proportional costs (lower right panel). The corresponding (daily) frictionless equilibrium return is constant and equal to $0.072$ here.}
    \label{plots1}
\end{figure}

\begin{figure}[htbp] 
    \centering
        \includegraphics[width=0.40\textwidth]{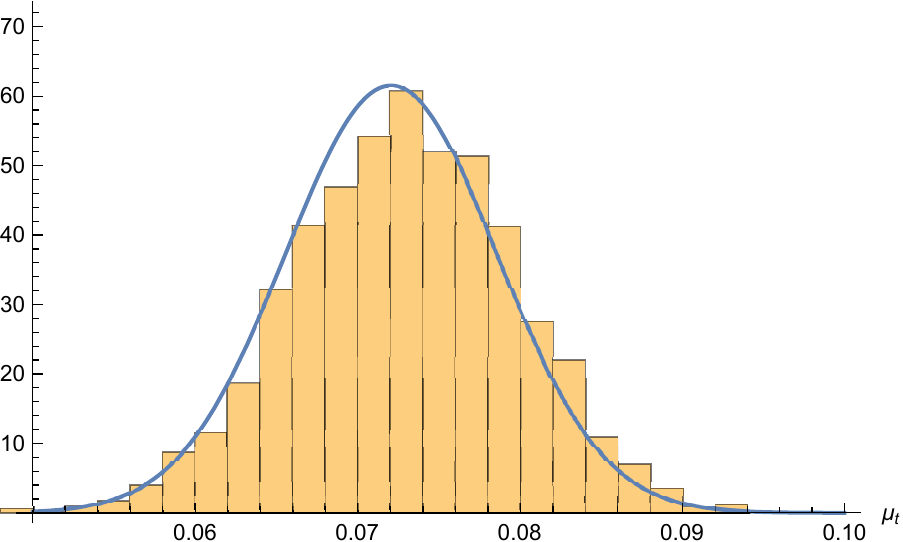}
         \includegraphics[width=0.40\textwidth]{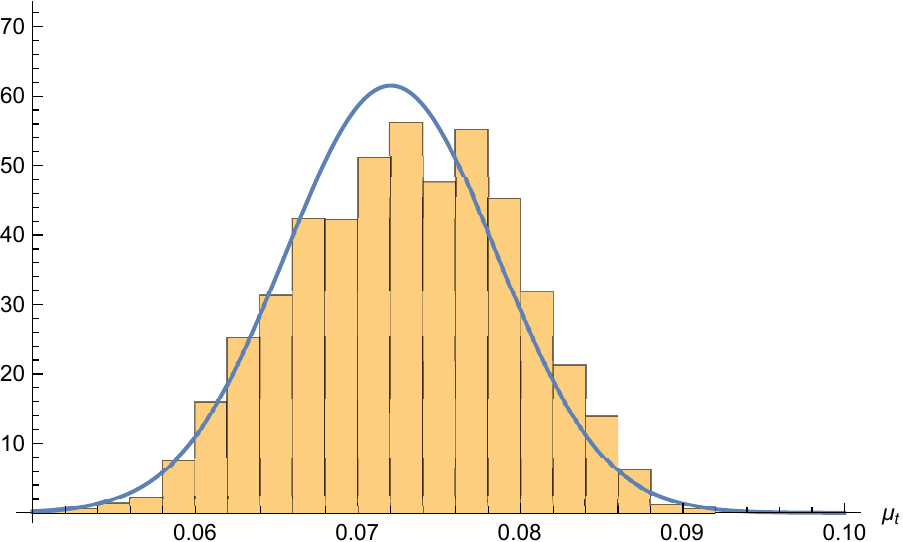}
        \includegraphics[width=0.40\textwidth]{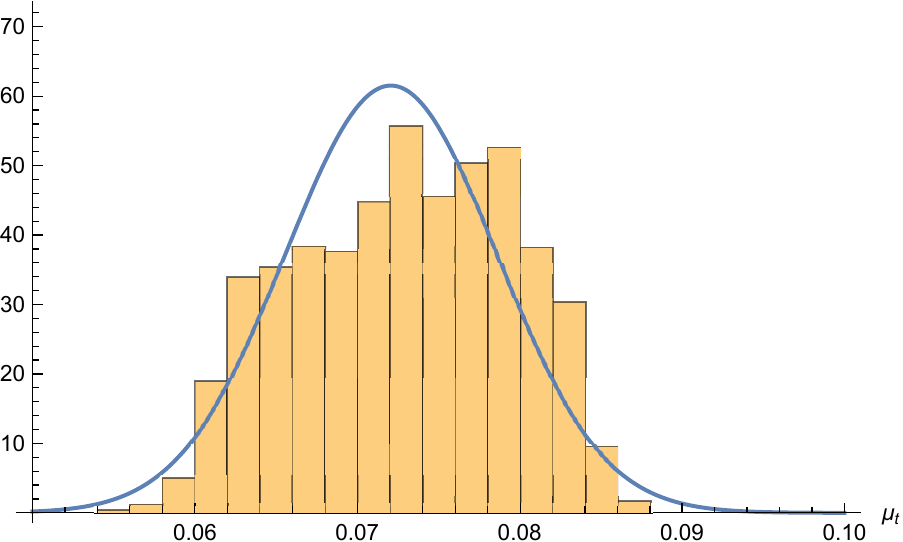}
        \includegraphics[width=0.4\textwidth]{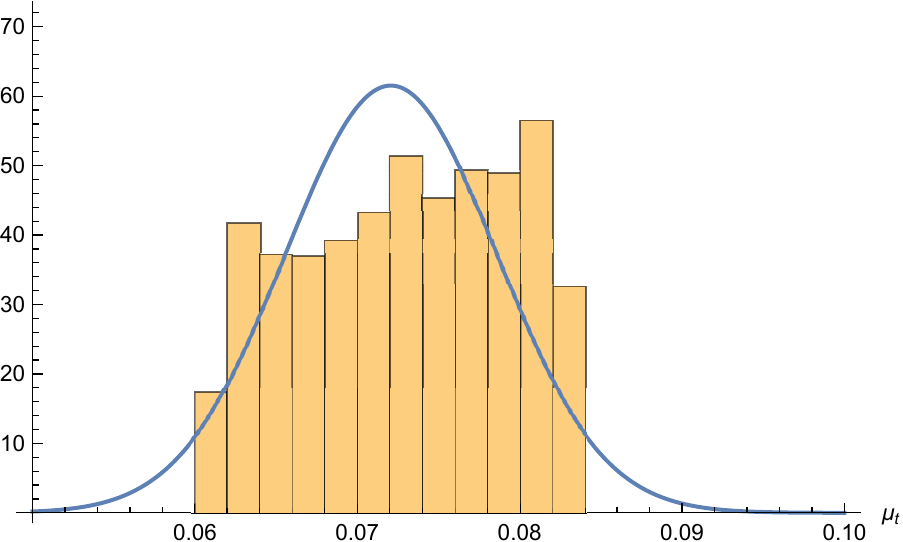}
    \caption{Empirical probability density functions of the simulated equilibrium returns for quadratic trading costs (left upper panel), costs proportional to the $3/2$-th power of the agents' trading rates (right upper panel), to the $9/8$-th power (lower left panel), and proportional costs (lower right panel) compared to the stationary normal distribution for quadratic costs.}
    \label{plots2}
\end{figure}

With proportional costs, this process has a uniform stationary law with standard deviation $l/\sqrt{3}$. With quadratic costs, the stationary standard deviation of the Ornstein-Uhlenbeck state variable is $\beta_2/\sqrt[4]{4\gamma/\lambda_2\sigma^2} = \text{ShTu}\sqrt{\pi\lambda_2}/\sqrt{2\gamma\sigma^2}$. To match this with the stationary standard deviation for proportional costs, we choose $\lambda_2=1.08\times 10^{-10}$. This in turn leads to $\beta_2 = 2.19\times 10^{10}$. Note that our ``equivalent quadratic cost'' is of the same order of magnitude as the direct estimates obtained from proprietary trade execution data for S\&P500 stocks in~\cite[Table~5]{collin.al.20}.

For general power costs $G_q(x)=\lambda_q |x|^q/q$ the solution $g_q$ of the ODE~\eqref{eqn:ergodic ODE} is not known explicitly. However, by exploiting the homotheticity of the power function, a change of variable allows to reduce \eqref{eqn:ergodic ODE} to an equation that only depends on the elasticity $q$ of the price impact function, but not the parameters $\lambda_q$, $\beta_q$ that we are trying to determine here. Accordingly, the values of $\lambda_q$, $\beta_q$ that match the average share turnover observed empirically as well as the variance of the state variable for proportional costs can be expressed as integrals of this universal function. For fixed $q$, these can in turn be computed by using a quadrature formula to integrate the numerical solution of~\eqref{eqn:ergodic ODE}, cf.~Appendix~\ref{app:calculation} for more details. For $q=1.5$ (which is in line with empirical estimates of actual trading costs in~\cite{almgren.al.05,lillo.al.03}), this leads to 
\begin{align*}
\beta_{1.5} = 2.33\times 10^{10}, \qquad \lambda_{1.5} = 5.22 \times 10^{-6}.
\end{align*}
Analogously, for $q=1.125$ (i.e., trading costs close to proportional), we obtain
\begin{align*}
\beta_{1.125} = 2.50\times 10^{10}, \qquad \lambda_{1.125} = 0.019.
\end{align*}

Simulations of ten years of daily equilibrium returns (generated with the same Brownian sample path to facilitate comparison) for these four sets of parameters are shown in Figure~\ref{plots1}. For our calibrated parameters, the frictional equilibrium returns display substantial deviations around their frictionless counterpart, but the differences between the equilibrium returns for the different cost specifications is much smaller.

Even though these numbers are generated from just one sample path, they in fact quite accurately reflect the stationary distributions of the state variables by the ergodic theorem. This is illustrated in Figure~\ref{plots2}, where we compare the empirical probability density functions to the stationary normal distribution for $q=2$. While the empirical distribution clearly does become more spread out for smaller $q$ (it is normal for $q=2$ but uniform for $q=1$), the realized distributions are nevertheless quite similar for our calibrated parameters. 

The simulated share turnover for $q=2$ and $q=3/2$ is compared to the historical trading volume data in Figure~\ref{plots3}. By the calibration of our model, the averages of the simulated trading volumes agree with the empirical data and broadly display the same mean-reverting behaviour. However, for the simple model with constant price volatility and homogenous trading needs, the variances of trading volume are substantially larger than in the data. Moreover, the autocorrelation functions in the model also decay much slower than their empirical counterparts. 

\begin{figure}[htbp] 
    \centering
        \includegraphics[width=0.4\textwidth]{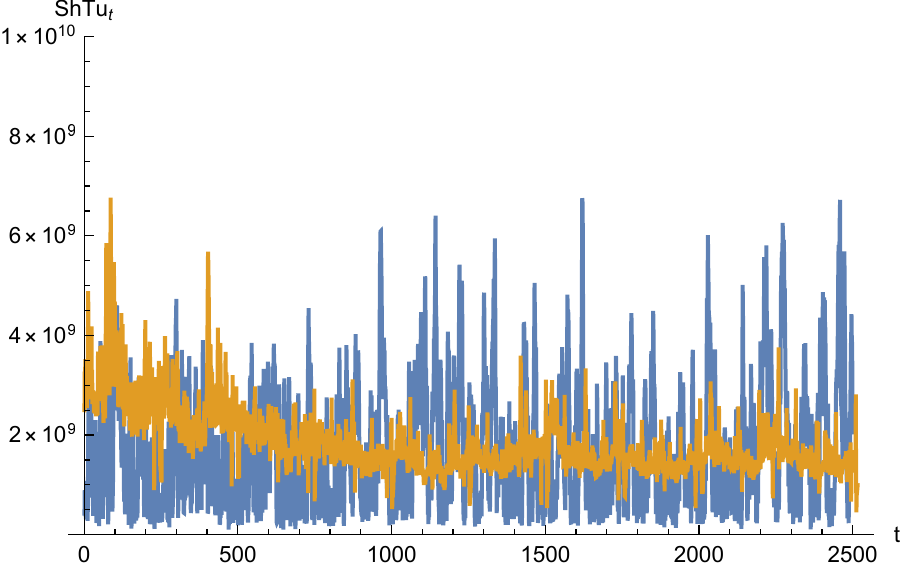}
         \includegraphics[width=0.4\textwidth]{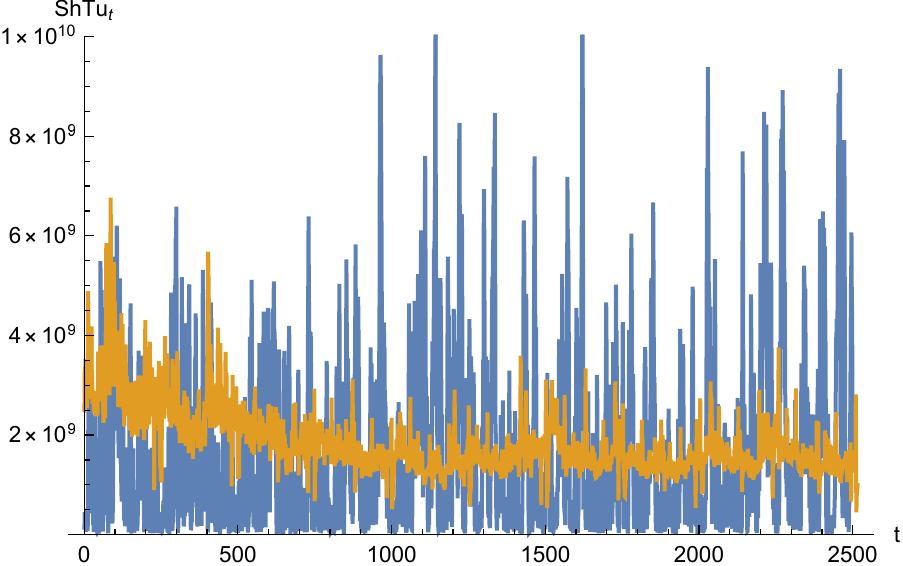}
    \caption{Simulated daily share turnover for quadratic trading costs (blue, left panel) and costs proportional to the $3/2$-th power of the agents' trading rates (blue, right panel), compared to empirical trading volume (orange).}
    \label{plots3}
\end{figure}

\section{More General Settings and Nonlinear FBSDEs}\label{sec:general}

We now discuss how the results from the previous sections formally extend to more general settings with a finite time horizon, more general state dynamics, and endogenous volatilities. Such extensions allow to address some of the shortcomings of our tractable baseline model. For example, more general state dynamics allow to generate stochastic volatility, an important stylized feature of price data in its own right as well as an important determinant of trading volume. Models with endogenous volatility open the door to studying how the latter depends on liquidity and can also generate the systematic liquidity premia observed empirically in expected returns. 

However, as we now outline, the analysis of such more general models leads to non-linear, fully-coupled systems of FBSDEs. For quadratic trading costs and sufficiently similar risk aversions of the two agents, some first wellposedness results have recently been developed in \cite{herdegen.al.19}. Extensions to more general trading costs are an intriguing but challenging direction for further research, but beyond our scope here. To nevertheless shed some first light on the qualitative properties of the equilibrium and the quantitative implications of our calibrated parameters, we discuss numerical algorithms based on the deep-learning approach of~\cite{han.al.17} in Section~\ref{ssec:numerics} below.

\subsection{Market}

In this section, we consider more general state dynamics than in Sections \ref{sec:frictionless}, \ref{sec:regular}, and \ref{sec:singular}. More specifically, for $n=1,2$, the cumulative endowment is of the general form
\begin{align*}
d\zeta^n_t =\beta^n_t dW_t, \quad \mbox{for a general $\beta^n \in \mathscr{L}^2$,}
\end{align*}
and the price of the risky asset has dynamics
\begin{align}\label{eq:dyn2}
dS_t=\mu_t dt+\sigma_t dW_t.
\end{align}
Now, not just the equilibrium return process $\mu \in \mathscr{L}^1$ but also the initial price $S_0 \in \mathbb{R}$ and the volatility process $\sigma \in \mathscr{L}^2$ are to be determined in equilibrium by matching the agents demand to the supply $s \in \mathbb{R}$ of the risky asset.  To pin down these additional quantities, we assume as in \cite{herdegen.al.19} that the terminal stock price is given by an exogenous $\mathscr{F}_T$-measurable random variable:
\begin{equation}\label{eq:temrinal2}
S_T=\mathfrak{S}.
\end{equation}
This can be interpreted as a fundamental value or as a terminal dividend.

\subsection{Frictionless Optimization and Equilibrium} \label{subsec:generalFrictionless}

The frictionless results from Section~\ref{s:fl} readily adapt this more general setting. Indeed, also for a general stochastic volatility process, pointwise maximization of the goal functional~\eqref{eq:nocosts} still yields the agents' individually optimal strategies,
\begin{align*}
\varphi^n_t = \frac{\mu_t}{\gamma^n \sigma_t^2} -\frac{\beta_t^n}{\sigma_t}, \quad t \in [0,T].
\end{align*}
The equilibrium return is then still pinned down by matching the agents' total demand $\varphi^1_t+\varphi^2_t$ to the supply $s$ of the risky asset:
\begin{equation}\label{eq:eqnocosts2}
\bar{\mu}_t= \bar{\gamma}\big[s\bar{\sigma}_t^2+\bar{\sigma}_t(\beta_t^1+\beta_t^2)\big],   \quad t \in [0,T], \quad \mbox{where } \bar{\gamma}=\frac{\gamma^1\gamma^2}{\gamma^1+\gamma^2}.
\end{equation}
Now, however, we also need to determine the corresponding initial price of the risky asset and its volatility. To this end, insert~\eqref{eq:eqnocosts2} into \eqref{eq:dyn2} and recall the terminal condition~\eqref{eq:temrinal2}. This leads the following scalar quadratic BSDE:
\begin{equation}\label{eq:bsde2}
d\bar{S}_t= \bar{\gamma}\big[s\bar{\sigma}_t^2+\bar{\sigma}_t(\beta_t^1+\beta_t^2)\big]dt+\bar{\sigma}_t dW_t, \quad \bar{S}_T=\mathfrak{S}.
\end{equation}
As is well known, the solution of this equation can be expressed in terms of the Laplace transform of the terminal condition, leading to explicit solutions in many concrete examples~\cite[Section~4.1]{herdegen.al.19}. 

\begin{example}\label{ex:bachelier}
If 
$$\beta^1+\beta^2=0 \quad \mbox{and} \quad \mathfrak{S}= bT+aW_T,$$ 
then the frictionless equilibrium price $\bar{S}$ is a Bachelier model with constant expected returns and volatilities:
$$
\bar{S}_t = (b-s\bar{\gamma}a^2)T+ s\bar{\gamma}a^2 t+a W_t, \quad t\in [0,T],
$$
Agents $n=1,2$'s optimal trading strategies in this frictionless equilibrium are
\begin{align}\label{frictionless strategies}
\bar\varphi^n_t = \frac{s\bar{\gamma}}{\gamma^n} -\frac{\beta_t^n}{a}, \quad t \in [0,T].
\end{align}
\end{example}

\subsection{Frictional Optimization and Equilibrium}

With transaction costs, both individual optimization and the corresponding equilibria become significantly more involved, leading to systems of fully-coupled nonlinear FBSDEs. Let us first consider the agents' individual optimization problems for a given initial asset price $S_0 \in \mathbb{R}$, expected returns process $(\mu_t)_{t\in [0,T]}$ and volatility process $(\sigma_t)_{t \in [0,T]}$. By strict convexity of the goal functional~\eqref{eq:goal2}, optimality of a trading rate $\dot{\varphi}^n$ for agent $n$ is equivalent to the first-order condition that the Gateaux derivative $\lim_{\rho \to 0} \frac{1}{\rho}(J^n_T(\dot\varphi^n+\rho\dot\varphi)-J^n_T(\dot\varphi^n))$ vanishes for \emph{any} perturbation $\varphi$, cf.~\cite{ekeland.temam.99}:
$$0=\mathbb{E}_t\left[\int_0^T \left(\mu_t \int_0^t \dot{\varphi}_u du-\gamma^n\sigma_t(\sigma_t\varphi^n_t+\beta^n_t) \int_0^t \dot{\varphi}_u du - G'(\dot{\varphi}^n_t)\dot{\varphi}_t\right)dt\right].$$
As in \cite{bank.al.17}, this can be rewritten using Fubini's theorem as 
$$
0=\mathbb{E}_t\left[\int_0^T \left(\int_t^T \Big(\mu_u -\gamma^n\sigma_u(\sigma_u\varphi^n_u+\beta^n_u)\Big) du -G'(\dot{\varphi}^n_t)\right)\dot{\varphi}_t dt\right].
$$
Since this has to hold for \emph{any} perturbation $\dot{\varphi}_t$, the tower property of conditional expectation yields
\begin{align}
G'(\dot{\varphi}^n_t) &=\mathbb{E}_t\left[\int_t^T \mu_u -\gamma^n\sigma_u\left(\sigma_u\varphi^n_u+\beta^n_u\right)du\right] = M_t - \int_0^t  \Big(\mu_u -\gamma^n\sigma_u \left(\sigma_u\varphi^n_u+\beta^n_u\right)\Big) du,\label{eq:bwopt1}
\end{align}
for a martingale $M^n=M^n_0+\int_0^\cdot Z^n_t dW_t$ that needs to be determined as part of the solution.  Solving for the dynamics of the agents' optimal trading rates would introduce the derivatives of the trading cost. Accordingly, it is preferable to instead work with the marginal trading cost as the backward process that describes the agents' optimal controls: 
\begin{align*}
Y^n_t = G'\left(\dot{\varphi}^n_t\right).
\end{align*}
With this notation, the corresponding trading rates are $\dot{\varphi}^n_t = (G')^{-1}(Y^n_t)$ and agent $n$'s optimal position $\varphi^n$ and the corresponding marginal trading costs $Y^n$ in turn solve the nonlinear FBSDE
\begin{align}
d\varphi^n_t &=(G')^{-1}(Y^n_t) dt, \qquad  &\varphi^n_0=\varphi^n_{0-},  \label{eq:eqphi}\\
dY^n_t &= \big(\gamma^n(\sigma_t\varphi^n_t+\beta^n_t)\sigma_t^\top -\mu_t\big)dt + Z^n_t dW_t, \qquad &Y^n_T =0.\label{eq:bwopt}
\end{align}
(Here, the terminal condition follows from $Y^n_T=G'(\dot{\varphi}^n_T)=G'(0)=0$.)  For constant quadratic costs $\lambda x^2/2$ and constant volatility $\sigma$, this FBSDE becomes linear and can in turn be solved by reducing it to some standard Riccati equations~\cite{bank.al.17,bouchard.al.18}. For volatilities and quadratic costs that fluctuate randomly, these ODEs are replaced by a backward \emph{stochastic} Riccati equation, compare~\cite{kohlmann.tang.02,annkirchner.kruse.15}. With nonlinear costs, no such simplifications are possible. In fact, the wellposedness of the system is generally unclear even for short time horizons since no Lipschitz condition for $(G')^{-1}$ is satisfied for costs of power form $G(x)= |x|^q/q$, $q \in (1,2)$, for example.

Despite these difficulties, formally solving for the corresponding equilibrium return is -- surprisingly -- not more difficult than for quadratic costs. To see this, first observe that symmetry of the trading cost $G$ implies that the marginal cost $G'$ and in turn its inverse $(G')^{-1}$ are antisymmetric. As a consequence, the market-clearing condition $\dot{\varphi}_t^1=-\dot{\varphi}_t^2$ implies that $(G')^{-1}(\dot{\varphi}^1_t)=-(G')^{-1}(\dot{\varphi}^2_t)$ and in turn
\begin{align}\label{eq:clear3}
Y^1_t + Y^2_t = 0, \qquad \mbox{for all}\; t\in[0,T].
\end{align}
After summing the corresponding backward equations~\eqref{eq:bwopt}, it follows that the frictional equilibrium return has to satisfy 
\begin{align*}
0 
&=  \mu_t -\gamma^1(\sigma_t\varphi^1_t+\beta^1_t)\sigma_t +  \mu_t -\gamma^2(\sigma_t\varphi^2_t+\beta^2_t)\sigma_t.
\end{align*}
Together with the market-clearing condition $\varphi^1_t+\varphi^2_t=s$, it follows that the frictional equilibrium return has the same relationship to the frictional volatility and the agents' optimal positions as for quadratic costs~\cite{herdegen.al.19}:
\begin{align}\label{mufric}
\mu_t 
&= \frac{1}{2} \left[(\gamma^2 s+(\gamma^1 - \gamma^2) \varphi^1_t)\sigma_t  +(\gamma^1\beta^1_t+\gamma^2\beta^2_t) \right]\sigma_t. 
\end{align}
Plugging expression~\eqref{mufric} back into agent 1's optimality condition~\eqref{eq:bwopt} in turn yields a backward equation that is linear in the optimal position, like for quadratic costs:\footnote{Note that these linear dynamics obtain here if this equation is expressed in terms of the marginal cost $Y^1_t=G'(\dot{\varphi}^1_t)$ rather than the trading rate $\dot{\varphi}^1_t$.}
\begin{align}\label{eq:eqbwd}
dY^1_t 
&=  \left(\gamma \sigma_t \varphi^1_t-\frac{\gamma^2 s}{2}\sigma_t - \frac{1}{2}(\gamma^2\beta^2_t - \gamma^1\beta^1_t)\right)\sigma_t dt+Z^1_tdW_t, \quad Y^1_T=0.
\end{align}
All nonlinearities are absorbed into the corresponding forward component,
\begin{equation}\label{eq:eqfwd}
d\varphi^1_t =(G')^{-1}(Y^1_t)dt,\quad  \varphi^1_{0}=\varphi^1_{0-}.
\end{equation}
If the volatility process $\sigma$ is not given exogenously, it needs to be determined from the terminal condition $\mathfrak{S}$. By plugging expression~\eqref{mufric} for the equilibrium return into the price dynamics~\eqref{eq:dyn2}, we obtain the following BSDE, which is coupled to the forward-backward system (\ref{eq:eqbwd}-\ref{eq:eqfwd}):
\begin{equation}\label{eq:eqbwd2}
dS_t= \frac{\sigma_t}{2}\Big[s\sigma_t\gamma^2+\gamma^1\beta^1_t+\gamma^2\beta^2_t+(\gamma^1-\gamma^2)\sigma_t\varphi^1_t\Big]dt+\sigma_t dW_t, \quad S_T=\mathfrak{S}.
\end{equation}
This is again the same equation as for quadratic costs~\cite{herdegen.al.19}. In particular, if both agents' risk aversions coincide ($\gamma^1=\gamma^2$), it decouples from the forward-backward system (\ref{eq:eqbwd}-\ref{eq:eqfwd}) and leads to the same equilibrium price as without transaction costs. For heterogenous but sufficiently similar risk aversions $\gamma^1 \approx \gamma^2$ and quadratic costs, it is shown in \cite{herdegen.al.19} that a solution of (\ref{eq:eqbwd}-\ref{eq:eqbwd2}) exists and identifies an equilibrium with transaction costs. However, the proof crucially exploits that with quadratic costs, the forward-backward system (\ref{eq:eqbwd}-\ref{eq:eqfwd}) for a given volatility process $(\sigma_t)_{t \in [0,T]}$ can be studied by means of the stochastic Riccati equation from~\cite{kohlmann.tang.02}. Establishing such results for more general trading costs -- where such tools are not available -- is a challenging direction for further research.

Here, let us just briefly sketch how the nonlinear FBSDE~(\ref{eq:eqbwd}-\ref{eq:eqbwd2}) reduces to a nonlinear ODE in the context of Section~\ref{sec:regular}, where the endowment volatilities $\beta^n_t=\beta^n W_t, n=1,2$ follow Brownian motions. Since the volatility process is exogenous and constant there, we don't have to deal with the second backward component~\eqref{eq:eqbwd2} and, moreover, can work with the state variable
\begin{align*}
X_t = \varphi^1_t-\frac{s\gamma^2}{\gamma^1+\gamma^2}+\frac{\gamma^1\beta^1_t-\gamma^2\beta^2_t}{(\gamma^1+\gamma^2)\sigma}.
\end{align*}
With this notation, the forward-backward system (\ref{eq:eqbwd}-\ref{eq:eqfwd}) becomes autonomous,
\begin{align}
&dX_t = (G')^{-1}\left(Y^1_t\right) dt + \frac{\gamma^1\beta^1-\gamma^2\beta^2}{(\gamma^1+\gamma^2)\sigma}dW_t, & X_0 &= \varphi^1_{0-} - \frac{s\gamma^2}{\gamma^1+\gamma^2}, \label{eq:fwd}\\
&dY^1_t = \frac{(\gamma_1+\gamma_2)\sigma^2}{2}X_t dt + Z^1_t dW_t, & Y_T &= 0. \label{eq:bwd}
\end{align}
Now use the standard ansatz that the backward component $Y^1_t$ should be a function $g(t,X_t)$ of time and the forward component. It\^o's formula and the dynamics of the forward component in turn yield
\begin{align*}
dY^1_t &= \left(g_t(t,X_t)+g_x(t,X_t)(G')^{-1}\left(g(t,X_t)\right)+\frac{1}{2}\left(\frac{\gamma^1\beta^1-\gamma^2\beta^2}{(\gamma^1+\gamma^2)\sigma}\right)^2g_{xx}(t,X_t) \right)dt
\\&\quad +\frac{\gamma^1\beta^1-\gamma^2\beta^2}{(\gamma^1+\gamma^2)\sigma}g_x(t,X_t)dW_t.
\end{align*}
Comparing the drift rate to the BSDE~\eqref{eq:eqbwd}, we therefore obtain the following semilinear PDE:
\begin{equation}\label{eq:semi}
g_t(t,x)+g_x(t,x)(G')^{-1}\left(g(t,x)\right)+\frac{1}{2}\left(\frac{\gamma^1\beta^1-\gamma^2\beta^2}{(\gamma^1+\gamma^2)\sigma}\right)^2g_{xx}(t,x)=  \frac{(\gamma_1+\gamma_2)\sigma^2}{2}x.
\end{equation}
For a long time horizon, the solution should become stationary ($g_t(t,x) \approx 0$). This leads to the nonlinear ODE from Lemma~\ref{ODE}:
\begin{align*}
\frac{1}{2}\left(\frac{\gamma^1\beta^1-\gamma^2\beta^2}{(\gamma^1+\gamma^2)\sigma}\right)^2 g''(x)+g'(x)(G')^{-1}\left(g(x)\right)= \frac{(\gamma_1+\gamma_2)\sigma^2}{2}x. \tag{\ref{eqn:ergodic ODE}}
\end{align*}
For finite time horizons, the PDE~\eqref{eq:semi} cannot be reduced to an ODE. Far from the terminal time $T$, it is natural to expect that the correct solution is still identified by the same growth condition as for the ODE~\ref{ODE}. For the numerical solution of the latter, the growth condition also provides a boundary condition that is approximately correct for large values of the space variable. For the PDE, however, this boundary condition in the space variable is incompatible with the zero terminal condition at maturity, which describes that trading slows down and eventually stops near the terminal time. For more general versions of the model, even the stationary boundary conditions in the space dimensions are not readily available and it is not clear how to paste them together with the terminal condition. Accordingly, it is not straightforward to solve the PDE~\eqref{eq:semi} and its extensions using finite-difference schemes.

As an alternative, in the next section we therefore propose a numerical algorithm in the spirit of \cite{han.al.17}. It solves the FBSDE by simulation and therefore bypasses the need to identify the correct boundary conditions. The algorithm approximates the dependence of the backward component on the forward components by a deep neural network. Whence, it is also able to handle higher-dimensional settings, e.g., with endogenous volatilities or random and time-varying transaction costs.

\section{Numerics}\label{ssec:numerics}

We now present a numerical algorithm to solve the FBSDEs from Section~\ref{sec:general}. The algorithm is then tested for the calibrated parameters from Section~\ref{sec:calibration}.

\subsection{Deep-Learning Algortihm}\label{ss:deep}

\paragraph{Overview}
Solving the forward-backward system is challenging because it is multidimensional and the forward and backward components are fully coupled. Nevertheless, it is amenable to the simulation-based approach of~\cite{han.al.17}, which approximates the solution by a deep neural network. In \cite{han.al.17} the focus lies on BSDEs but the approach can readily be extended to FBSDEs, compare~\cite{han2018convergence}.

Let us briefly sketch the main idea; further details on the implementation are provided below. The first step is to pass to a time-discretized version of (\ref{eq:eqbwd}-\ref{eq:eqbwd2}), e.g., using the Euler scheme. Solving this system amounts to finding at each time step $t_k$ the unknown ``controls'' $Z^1_{t_k}$, $\sigma_{t_k}$. If the terminal condition is a function $\mathfrak{S}(W_T)$ of the underlying Brownian motion only as in Example~\ref{ex:bachelier}, then it is well known that the solution and in turn $Z^1_{t_k}$, $\sigma_{t_k}$ are functions of the forward variables, $Z^1_{t_k}=F^{\theta_k^Z}(W_{t_k}$, $\varphi^1_{t_k})$ and $\sigma_{t_k}=F^{\theta_k^\sigma}(W_{t_k}$, $\varphi^1_{t_k})$.

The algorithm of~\cite{han.al.17} approximates each of these functions with a function in the class $\{ F^{\bar{\theta}} \colon \bar{\theta} \in \Theta \}$ of neural networks, where we write $\theta=(\theta_0^Y,\theta_0^S,\theta_0^\sigma,\ldots,\theta_n^\sigma,\theta_0^Z,\ldots,\theta_n^Z)$ for the collection of all the corresponding parameters. The goal now is to choose these parameters in order to match the terminal conditions $Y^1_T = 0$ and $S_T = \mathfrak{S}(W_T)$ of the system sufficiently well. To this end, one starts with an initial guess for the network functions and then simulates the system forward in time. In this way, a simulated Brownian sample path is mapped to a corresponding terminal condition. This mapping can be efficiently implemented as a computational graph, which is determined by the choice of the building block networks $\{ F^{\bar{\theta}} \colon \bar{\theta} \in \Theta \}$ (i.e., two networks of type~\eqref{eq:NNarchitecture} for each time-step) and the FBSDE system, which describes how these building block networks are concatenated over time (see \eqref{eq:eulerFBSDE} below). To iteratively update the network functions until the terminal conditions are matched sufficiently well, one may then leverage computational technology available for such networks, such as backpropagation and stochastic gradient descent-type algorithms, see e.g.\ \cite[Chapters~6 and 8]{Goodfellow2016}. This can be implemented efficiently, e.g., in Python using Tensorflow. 

\paragraph{Algorithm}

Let us now describe the approximation algorithm in more detail. Fix a discrete time grid $0=t_0 < t_1 < \ldots < t_{n} = T$. For any choice of parameter $\theta$, consider the following discrete-time forward system obtained by discretizing (\ref{eq:eqbwd}-\ref{eq:eqbwd2}): starting from initial values $\varphi_0^{1,\theta} = \varphi_{0-}^1$, $Y_0^{1,\theta} = \theta_0^Y$, $S_0^\theta=\theta_0^S$, for $k=0,\ldots,n-1$ calculate
 \begin{equation} \label{eq:approxControl} 
\begin{aligned}
Z_{t_k}^{1,\theta} & = F^{\theta_k^Z}(W_{t_k}, \varphi_{t_k}^{1,\theta}), \quad \sigma_{t_k}^\theta & = F^{\theta_k^\sigma}(W_{t_k}, \varphi_{t_k}^{1,\theta}), 
\end{aligned}
\end{equation}
and step forward according to the Euler scheme
\begin{align} 
\varphi_{t_{k+1}}^{1,\theta} & = \varphi_{t_{k}}^{1,\theta} + (G')^{-1}(Y_{t_k}^{1,\theta})(t_{k+1}-t_k), \label{eq:eulerFBSDE} \\
Y_{t_{k+1}}^{1,\theta} & = Y_{t_k}^{1,\theta}  +\frac{\sigma_{t_k}^\theta}{2} \left[ \sigma_{t_k}^\theta \varphi_{t_k}^{1,\theta} (\gamma^1+\gamma^2) - \sigma_{t_k}^\theta s \gamma^2 + \gamma^1 \beta^1_{t_k} - \gamma^2 \beta^2_{t_k} \right] (t_{k+1}-t_k) + Z_{t_k}^{1,\theta} (W_{t_{k+1}}-W_{t_k}),	 \notag \\
S_{t_{k+1}}^\theta & = S_{t_{k}}^\theta + \frac{\sigma_{t_k}^\theta}{2} \left[ s \sigma_{t_k}^\theta \gamma^2 + (\gamma^1-\gamma^2) \sigma_{t_k}^\theta \varphi_{t_k}^{1,\theta} + \gamma^1 \beta^1_{t_k} + \gamma^2 \beta^2_{t_k} \right](t_{k+1}-t_k)  + \sigma_{t_k}^\theta (W_{t_{k+1}}-W_{t_k}). \notag
\end{align}
For any choice of the approximation parameter $\theta$, this defines a discrete-time stochastic process, but of course the terminal conditions $Y_T^{1,\theta} = 0$ and $S_T^\theta = \mathfrak{S}(W_T)$ will not even be approximately satisfied for an arbitrary choice of $\theta$. However, if $\hat{\theta}$ is a minimizer of
\begin{equation} \label{eq:mintheta}
\min_{\theta} \, \mathcal{L}(\theta), \quad \mbox{where $\mathcal{L}(\theta) = \E[(Y^{1,\theta}_T)^2] + \E[(S_T^\theta - \mathfrak{S})^2]$},
\end{equation}
where the number $n$ of time steps is sufficiently large and the function class $\{ F^{\bar{\theta}} \colon \bar{\theta} \in \Theta \}$ is sufficiently rich, then $(\varphi^{1,\hat{\theta}},Y^{1,\hat{\theta}},S^{\hat{\theta}}) $ should be a good approximation for the solution $(\varphi^1,Y^1,S)$ of (\ref{eq:eqbwd}-\ref{eq:eqbwd2}) at the time-points $t_0,\ldots,t_n$.

The minimization problem~\eqref{eq:mintheta} can be tackled using the ``stochastic gradient descent algorithm''. The main idea is the following: if the objective functional $\mathcal{L}$ was known explicitly and differentiable, then the classical gradient descent algorithm could be applied. That is, starting from an initial guess $\theta^{(0)}$, one iteratively updates
\begin{equation}\label{eq:update}
\theta^{(j+1)} = \theta^{(j)} - \eta_j \nabla \mathcal{L}_j(\theta^{(j)}),
\end{equation}
where $\mathcal{L}_j = \mathcal{L}$ and the learning rate $\eta_j > 0$ is fixed ($\eta_j = \eta$ for all $j$) or decreasing to $0$. Under suitable assumptions on $\mathcal{L}$ and $\{\eta_j \}_{j \in \mathbb{N}}$ the parameter $\theta^{(j)}$ then converges to a (local) minimum of $\mathcal{L}$ as $j \to \infty$. However, since $\mathcal{L}$ is not known explicitly, one applies the \textit{stochastic} gradient descent algorithm, which is the same procedure as just described, but approximates the expectations in $\mathcal{L}$ by a sample average in each iteration $j$,
$$
\mathcal{L}_j(\theta) = \frac{1}{N_{b}} \sum_{i=1}^{N_b}\left[ (Y_T^{1,\theta}(W^i))^2 + (S_T^\theta(W^i) - \mathfrak{S}(W^i))^2\right].
$$
Here, $N_{b} \in \mathbb{N}$ is called the ``batch size'' and $Y_T^{1,\theta}(W^i)$, $S_T^\theta(W^i)$ are calculated by plugging independent Brownian motions $W^1,\ldots,W^{N_b}$ into the Euler scheme~(\ref{eq:approxControl}-\ref{eq:eulerFBSDE}).

In order to apply the updating rule~\eqref{eq:update}, one needs to be able to calculate $\nabla \mathcal{L}_j(\theta)$ efficiently and this is the point at which the choice of $\{ F^{\bar{\theta}} \colon \bar{\theta} \in \Theta \}$ becomes crucial. As is apparent from (\ref{eq:approxControl}-\ref{eq:eulerFBSDE}), the dependence of the solution on the parameter $\theta$ is complex, since the state variables and parametric functions are iteratively added, multiplied and composed. For instance $Z_{t_k}^{1,\theta}$ depends not only on $\theta_k^Z$, but also (via $\varphi_{t_k}^{1,\theta}$) on $\theta_{0}^Z,\ldots,\theta_{k-2}^Z$ and $\theta_0^\sigma, \ldots, \theta_{k-2}^\sigma$. This makes the computational solution of \eqref{eq:mintheta} by classical numerical techniques highly challenging. Whence, while in principle any sufficiently rich parametric family of functions could be chosen for $\{ F^{\bar{\theta}} \colon \bar{\theta} \in \Theta \}$ in the scheme described above, it turns out to be particularly useful to choose a class of neural networks here. Then, $Y^{1,\theta}_T$ and $S_T^\theta$ can be viewed as the outputs of a deep neural network with random input $(W_{t_k})_{k=0,\ldots,n}$. Thanks to the compositional structure of neural networks one can then use the chain rule to calculate the gradient $\nabla \mathcal{L}_j(\theta)$ in closed form. Furthermore, the resulting analytical expressions can be evaluated efficiently using the so-called backpropagation algorithm, see, e.g., \cite{Goodfellow2016}. By using subgradients, this also extends to e.g.\ the ``ReLU activation function'' used below. Finally, all of this can be implemented efficiently in the computational graph structure employed in libraries such as Tensorflow or Torch.  

In summary, the learning algorithm iteratively updates the network parameters $\theta$ until a desired approximation accuracy is reached for some $\hat{\theta}$. Note that the accuracy of the approximation can be verified out of sample (e.g., in the numerical experiments in in Section~\ref{subsec:numericalResults}) by simulating a large number $N_{\mathrm{test}}$ of additional independent sample paths of $W$ and evaluating the empirical loss $\mathcal{L}_j(\hat{\theta})$ (with $N_b=N_{\mathrm{test}}$) on this collection of test paths.

\paragraph{Implementation}

For the numerical experiments in Section~\ref{subsec:numericalResults}, each $F^{\bar{\theta}}$ is a  neural network with two hidden layers. For the activation function we choose the popular Rectified Linear Unit (\emph{ReLU}) $\bm{\rho}$, which applies $x \mapsto \max(x,0)$ to each component of a vector. Denoting by $N_1,N_2 \in \mathbb{N}$ the number of nodes in the hidden layer, we thus consider functions of the form 
\begin{equation} \label{eq:NNarchitecture} F^{\bar{\theta}}(x)=A^2 \bm{\rho}(A^1 \bm{\rho}(A^0 x + b^0) + b^1), \quad x \in \R^2, \end{equation}
where $A^0 \in \R^{N_1\times 2}$, $b^0 \in \R^{N_1}$, $A^1 \in \R^{N_2 \times N_1}$, $b^1 \in \R^{N_2}$, $A^2 \in \R^{1 \times N_2}$ are called the \emph{weights} and \emph{biases} of the network and we denote by $\Theta$ the set of all parameters $\bar{\theta}=(A^0,b^0,A^1,b^1,A^2)$. To find a close-to-optimal parameter $\hat\theta$ in \eqref{eq:mintheta} we randomly initialize the network parameters and subsequently use the Adam algorithm \cite{Ioffe2015}, which is a variant of stochastic gradient descent which adaptively adjusts the learning rates for all network parameters. Here, some initial hyperparameter optimization has led us to choose $N_1=N_2=15$, set the initial learning rate to $0.0005$ and use a batch size of $128$. In order to accelerate the parameter training procedure, we apply batch-normalization \cite{Kingma2015} (see also \cite[Section~8.7.1]{Goodfellow2016}) at different stages: before the input is fed into the network, before applying the activation function $\bm{\rho}$ and after the last linear transformation $A^2$. All computations are performed in Python using Tensorflow. 

\subsection{Numerical results}\label{subsec:numericalResults}

The algorithm introduced in Section~\ref{ss:deep} is now applied to solve the forward-backward equations corresponding to Example~\ref{ex:bachelier}. As a sanity check, we first consider the simplest version of the model, where the price volatility is exogenous. In this setting, we compare the numerical solution to the nonlinear ODE that describes the exact solution of the infinite-horizon version of the model. 

Subsequently, we consider the model with endogenous volatility. In order to test the performance of the learning algorithm in this case, we compare its results to the semi-explicit solution in term of Riccati equations obtained for quadratic costs in~\cite{herdegen.al.19}.

\paragraph{Exogenous volatility}
We first consider the finite-horizon version of the model from Section~\ref{sec:regular} with power costs $G_q(x)=\lambda_q |x|^q/q$, where $q=1.5$, $\varphi^1_{0} = \frac{\gamma^2}{\gamma^1+\gamma^2}s=s-\varphi^2_{0-} $ and the model parameters are calibrated as in Section~\ref{sec:calibration}. The algorithm described in Section~\ref{ss:deep} for the general FBSDE (\ref{eq:eqbwd}-\ref{eq:eqbwd2}) can be readily adapted by setting $\sigma^\theta_{t_k} = \sigma$ for all $k$, only considering the first two equations in \eqref{eq:eulerFBSDE}, and minimizing $\E[(Y_T^{1,\theta})^2]$. An alternative, slightly more efficient approach is to use instead the system \eqref{eq:fwd}, \eqref{eq:bwd} and discretize it analogously to \eqref{eq:eulerFBSDE}. The algorithm from Section~\ref{ss:deep} in turn yields a parameter $\hat\theta$ such that $(X^{\hat\theta},Y^{1,\hat\theta})$ approximately solves (\ref{eq:fwd}-\ref{eq:bwd}). This allows us to generate approximate samples of (\ref{eq:fwd}-\ref{eq:bwd}) by simulating sample paths of $W$ and evaluating $(X^{\hat\theta},Y^{1,\hat\theta})$. On the other hand we know that $Y^1_{t_k} = g(t_k,X_{t_k})$, where $g$ solves \eqref{eq:semi}.  Thus we generate $N_{\mathrm{test}} = 10^6$ samples of $W$, evaluate $(X^{\hat\theta},Y^{1,\hat\theta})$ on each of them and obtain an approximation $\hat{g}(t_k,x)$ of $g(t_k,x)$ by assigning to each point $x$ which is attained by a sample of $X^{\hat\theta}_{t_k}$ the associated sample of $Y^{1,\hat\theta}_{t_k}$. This yields an approximation of the solution to \eqref{eq:semi} on a (random) grid specified by the state variable. According to \eqref{eq:eqfwd}, the corresponding optimal trading rate is in turn obtained by applying $(G_q')^{-1}(\hat{g}(t_k,\cdot))$ to the state variable.

We now compare this to the long-run optimal trading rate from Theorem~\ref{thm:regular}, where $g$ is given by the solution of the nonlinear ODE from Lemma~\ref{ODE}. Figure~\ref{fig:ODEcomparison} shows the graph of both functions at $t=T-t_k=25$, i.e., the samples of $(X_{t_k}^{\hat\theta},(G_q')^{-1}(Y_{t_k}^{1,\hat\theta}))=(X_{t_k}^{\hat\theta},(G_q')^{-1}(\hat{g}(t_k,X_{t_k}^{\hat\theta})))$ and $(X_{t_k}^{\hat\theta},(G_q')^{-1}(g(X_{t_k}^{\hat\theta})))$. We observe that the long-run optimum is already very close to the numerical-solution of the finite-horizon problem even for a time horizon of just five weeks. On the one hand, this justifies the use of the long-run model as a tractable approximation of its finite-horizon counterpart. On the other hand, it demonstrates that the deep learning algorithm indeed converges to the correct solution in this simplest version of the model.

\begin{figure}[htbp] 
    \centering
        \includegraphics[width=0.45\textwidth]{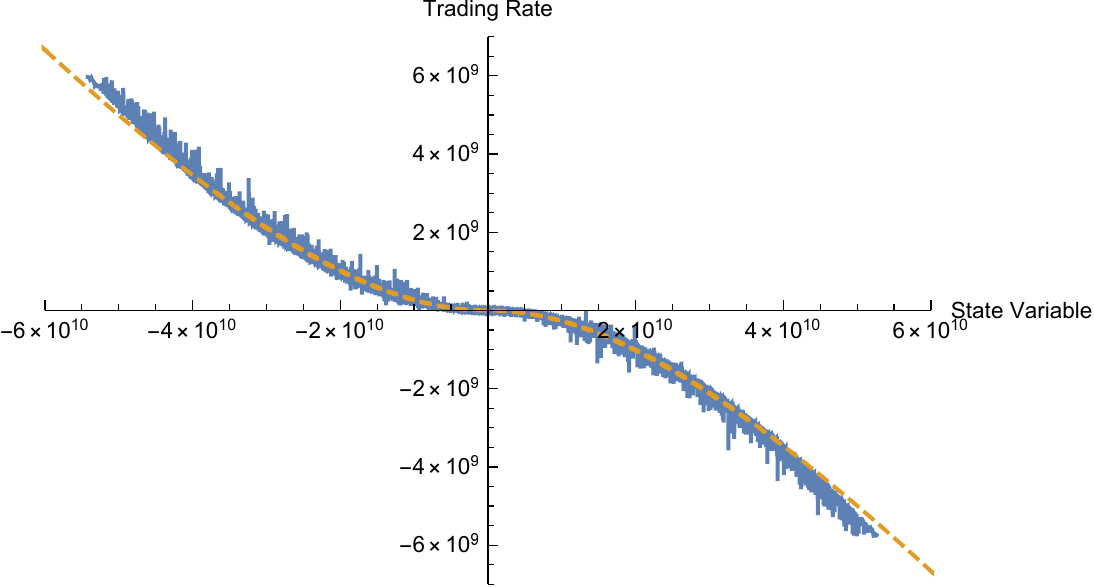} 
   \caption{Long-run optimal trading rate for power costs with $q=1.5$ (dashed) and the neural-network approximation of its finite-horizon counterpart 25 days before the terminal time (solid).}
    \label{fig:ODEcomparison}
\end{figure}

\paragraph{Endogenous volatility}

We now turn to the model with endogenous volatility from Section~\ref{sec:general}. We consider $G_q(x)=\lambda_q |x|^q/q$ both for $q=2$ (quadratic costs) and $q=1.5$ (power costs). For $\lambda_q$, $\gamma_1$, $\gamma_2$, and $\beta^1=-\beta^2=\beta_q$ we use the same parameter values as for the model with exogenous volatility (cf.~Section~\ref{sec:calibration}) and we also again set $\varphi^1_{0-} = \frac{\gamma^2}{\gamma^1+\gamma^2}s=s-\varphi^2_{0-}$. The additional parameters $a$ and $b$ are calibrated to the frictionless equilibrium from Section~\ref{subsec:generalFrictionless}. To wit, $a$ is estimated from the time series (resulting in the same value as for $\sigma$ in Section~\ref{subsec:calibrationEx}) and $b$ is chosen so that $\bar{S}_0=(b-s\bar{\gamma}a^2)T$ matches the current stock price. We focus on a short time horizon $T=20$ discretized into $n=40$ time steps. 

The deep-learning algorithm from Section~\ref{ss:deep} in turn yields an approximate solution of the forward-backward system (\ref{eq:eqbwd}-\ref{eq:eqbwd2}). To assess the effect of different  transaction costs we compare the equilibrium price and volatility to the respective quantities in the frictionless equilibrium, i.e., we examine (sample paths of) the price difference $S^{\hat{\theta}}-\bar{S}$ and the volatility difference $\sigma^{\hat{\theta}} - a$ over time. For quadratic costs it has been shown in \cite{herdegen.al.19} that optimal trading rates and the equilibrium price can be described in terms of a system of coupled Riccati ODEs, which provides a natural benchmark in this case. Figure~\ref{plots2a} shows two sample paths of the price and volatility corrections for quadratic costs calculated by both methods, i.e., by applying the neural network based algorithm described above and by solving the system of ODEs derived in \cite{herdegen.al.19} using a standard ODE solver. We see that the neural network based method provides a very accurate approximation of the equilibrium quantities. Note that since both the initial price correction and the volatility correction are deterministic here, this illustrates the accuracy of the approximation even though we only plot a few sample paths for illustrative purposes.

 \begin{figure}[htbp] 
	\centering
	\includegraphics[width=0.45\textwidth]{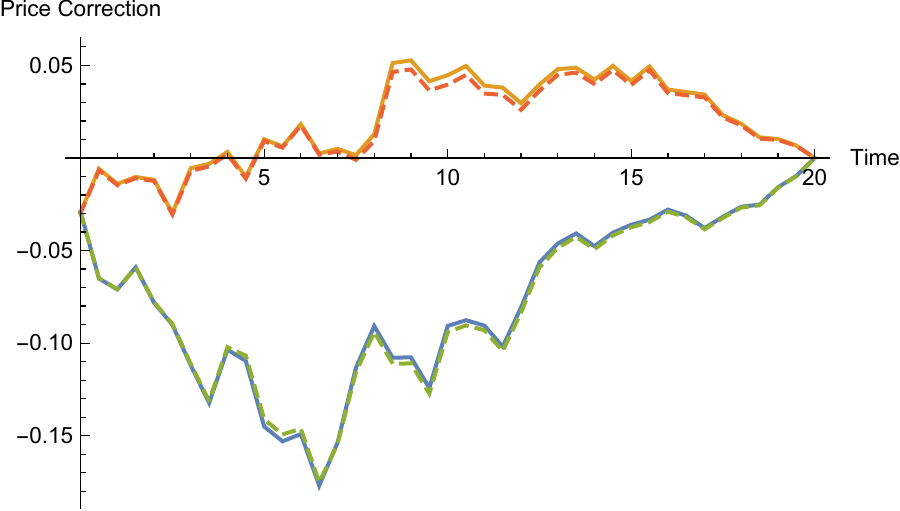}
	\includegraphics[width=0.45\textwidth]{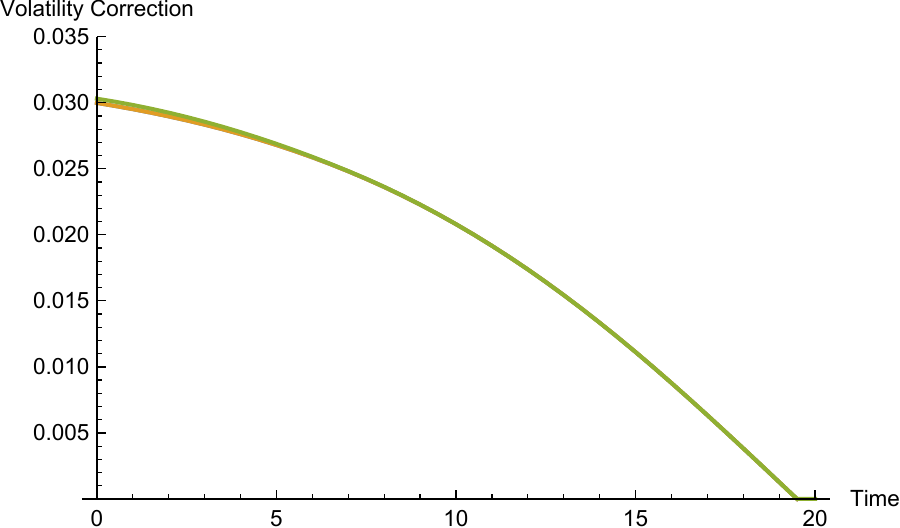}
	\caption{Price adjustment (left panel) and volatility adjustment (right panel) with calibrated parameters for quadratic costs.}
	\label{plots2a}
\end{figure}

Let us briefly comment on the size of the effects observed here. Since the frictionless and frictional equilibrium prices have the same terminal condition, the negative price initial correction plotted in the left panel of Figure~\ref{plots2a} corresponds to a positive liquidity premium, i.e., higher average expected returns in the presence of transaction costs as documented empirically in \cite{amihud.mendelson.86,brennan.subrahmanyam.96,pastor.stambaugh.03}, for example. The absolute price correction shown here corresponds to a (yearly) liquidity premium of about $0.3\%$ relative to the average stock prices. Whence, if one agent has twice the risk aversion of the one of the other as in our calibration, then matching the average trading volume observed empirically leads to a liquidity premium of the same order of magnitude as the ``equivalent proportional cost'' from Section~\ref{sec:calibration} ($0.25\%$). This is roughly in line with the empirical results of \cite{amihud.mendelson.86}, who find a ratio of $1.9$ and motivates our choice of $\gamma^2=2\gamma^1$. For this choice of parameters, the corresponding initial volatility is increased by about $1.6\%$ relative to its frictionless value.

The analogous plots for power costs with $q = 1.5$ are displayed in Figure~\ref{plots4}. (In order to make them comparable to their counterparts for quadratic costs, we use the same Brownian paths). Note that no benchmark is available in this case, but the equilibrium prices for the two cost specifications turn out to be very similar. To wit, the (yearly) liquidity premium for our matched power costs is about $0.32\%$ of the average stock prices, and the frictionless volatility initially increases by about $1.8\%$ of its frictionless value in this case.

\begin{figure}[htbp] 
	\centering
	\includegraphics[width=0.45\textwidth]{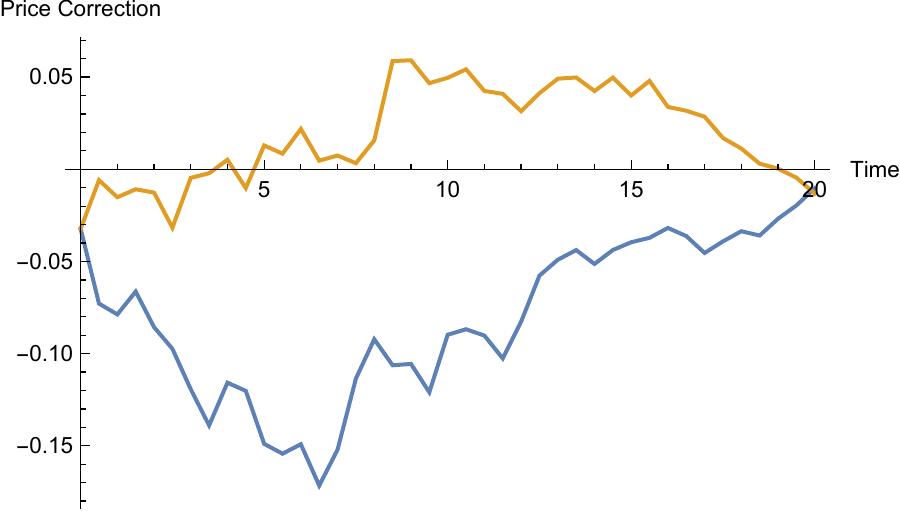}
	\includegraphics[width=0.45\textwidth]{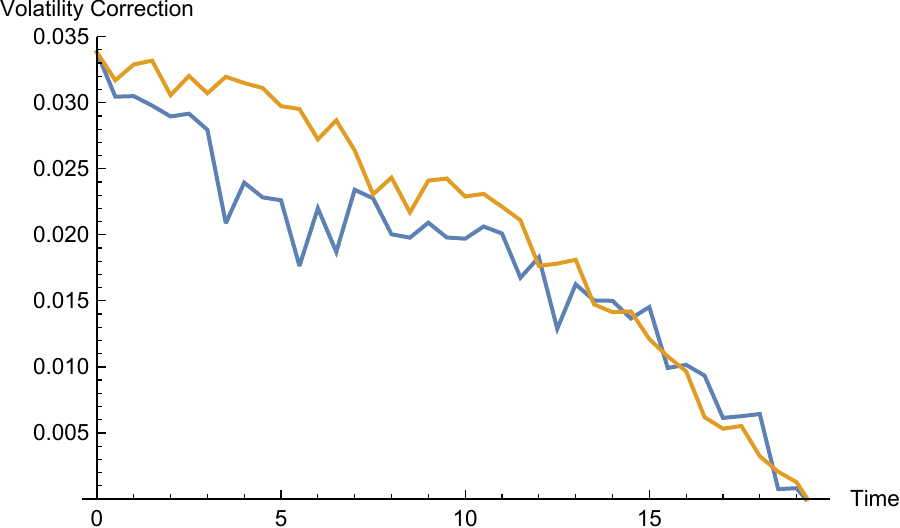}
	\caption{Price adjustment (left panel) and volatility adjustment (right panel) with calibrated parameters for $3/2$-costs.}
	\label{plots4}
\end{figure}

Overall, the numerical results reported in this section corroborate the findings from Section~\ref{sec:calibration} and suggest that quadratic costs can serve as useful proxies for other less tractable costs specifications also in settings with endogenous volatilities.

\section{Proofs}\label{s:proofs}
To ease notation, define 
\begin{align}\label{constants}
\gamma = \frac{\gamma^1+\gamma^2}{2}, \qquad \delta = \frac{\gamma^1\beta^1 - \gamma^2\beta^2}{(\gamma^1+\gamma^2)\sigma}.
\end{align}
Then the ODE~\eqref{eqn:ergodic ODE} in Lemma~\ref{ODE} can be rewritten as
\begin{align*}
\frac{\delta^2}{2} g''(x) + g'(x) \left(G'\right)^{-1}\left(g(x)\right) = \gamma\sigma^2 x,
\end{align*}
and the SDE~\eqref{eqn:ergodic SDE} in Lemma~\ref{lem:sde} rewrites as
\begin{align*}
dX_t = \left(G'\right)^{-1}\left(g(X_t)\right) dt + \delta dW_t.
\end{align*}
\subsection{Proofs for Section~\ref{sec:regular}}\label{sec:proofsregular}
\begin{proof}[Proof of Lemma~\ref{lem:sde}]
Strong existence and uniqueness follow from a standard localization argument, cf.~\cite[Proof of Proposition~1.1]{caye.al.20}.  That $X$ is a recurrent diffusion is established in~\cite[Appendix D.2]{caye.al.20}. For later use, we now also establish some uniform moment bounds for $X$ . To this end observe that by Lemma~\ref{ODE}, we have $g(x)\leq 0$ for $x>0$ and, in view of Assumption~\ref{cond:cost}(ii) there exists $K>0$ such that
$$
|(G')^{-1}(x)| \geq \frac{c}{2} |x| \quad \mbox{for $|x|\geq|K|$.}
$$
As $(G')^{-1}$ is odd, it follows that, for $x$ such that $|g(x)| \geq K$,
$$
x(G')^{-1}(g(x)) = -|x| (G')^{-1}(|g(x)|) \leq -\frac{c}{2} |x||g(x)|.
$$
Notice that $|g|$ is increasing on $[0,\infty)$ and satisfies $\lim_{|x|\to\infty} |g(x)| = \infty$. Whence, there exists $M_0>0$ such that for every $r>0$ and $|x|\geq 2r/c|g(M_0)|+M_0$,
$$
\frac{x(G')^{-1}(g(x))}{|x|} \leq  -\frac{c}{2}|g(x)| \leq -\frac{c}{2}|g(M_0)|\leq -\frac{r}{|x|}.
$$
Thus, \cite[Condition~(6)]{veretennikov1997polynomial} is satisfied and~\cite[Lemma~1]{veretennikov1997polynomial} is applicable for every $r>0$. Therefore, we have the following uniform moment bounds: 
\begin{align}\label{sde:moments}
\sup_{T\geq 0} \E\left[|X_T|^k\right] <\infty, \quad \mbox{for every $k\in \mathbb{N}$.}
\end{align}
\end{proof}

\begin{proof}[Proof of Theorem~\ref{thm:regular}]
Market clearing evidently holds by definition of the trading rates~\eqref{strategy}. Observe that the corresponding positions $\varphi^1$ satisfy
\begin{align}
\label{position:regular}
\mu_t - \gamma^1\sigma( \sigma\varphi^1_t + \beta^1_t) = -\gamma\sigma^2 X_t.
\end{align}
Consider a competing admissible strategy $\varphi$ for agent 1. Identity~\eqref{position:regular} and the convexity of $G$ yield
\begin{align}
&J_T^1(\dot{\varphi}) - J_T^1(\dot{\varphi}^1) \notag\\
&\quad= \E\left[ \int_0^T\big(\varphi_t - \varphi^1_t \big)\mu_t - 
\frac{\gamma^1}{2} \sigma\big(\varphi_t - \varphi^1_t \big)(\sigma\varphi_t + \sigma\varphi^1_t +2\beta^1_t) 
+G(\dot{\varphi}^1_t)-G(\dot{\varphi}_t) \; dt \right] \notag\\
&\quad = \E\left[ \int_0^T\big(\varphi_t - \varphi^1_t \big)\mu_t - 
\frac{\gamma^1}{2} \sigma\big(\varphi_t - \varphi^1_t \big)\big(\sigma\varphi_t- \sigma\varphi^1_t +2(\sigma\varphi^1_t + \beta^1_t)\big) 
+G(\dot{\varphi}^1_t)-G(\dot{\varphi}_t) \; dt \right] \notag\\
&\quad\leq \E\left[ \int_0^T-\frac1 2 \gamma^1 \sigma^2\big(\varphi_t - \varphi^1_t \big)^2 + \big(\mu_t - \gamma^1\sigma(\sigma\varphi^1_t + \beta^1_t)\big)\big(\varphi_t - \varphi^1_t \big) +G'(\dot{\varphi}^1_t)\big( \dot{\varphi}^1_t -\dot{\varphi_t} \big) \; dt\right] \notag\\
&\quad= \E\left[ \int_0^T-\frac1 2 \gamma^1 \sigma^2\big(\varphi_t - \varphi^1_t \big)^2 - \gamma \sigma^2 X_t\big(\varphi_t - \varphi^1_t \big) - G'(\dot{\varphi}^1_t)\big(\dot{\varphi}_t - \dot{\varphi}^1_t \big) \; dt\right]. \label{eq:comp}
\end{align}
We now analyze the terms on the right-hand side. To ease notation, set 
$$
\dot{\theta}_t = \dot{\varphi}_t - \dot{\varphi}^1_t, \qquad \mbox{so that} \quad \theta_t = \int_0^t \left(\dot{\varphi}_u - \dot{\varphi}^1_u\right) du =  \varphi_t - \varphi^1_t. 
$$
The dynamics~\eqref{eqn:ergodic SDE} of $X$, It\^o's formula, and the ODE~\eqref{eqn:ergodic ODE} for $g$ imply
\begin{align}
dg(X_t) 
&= \left[ \frac1 2 \delta^2 g''(X_t) + g'(X_t) (G')^{-1}(g(X_t)) \right] dt + \delta g'(X_t)dW_t  \notag\\
&= \gamma \sigma^2 X_t dt + \delta g'(X_t)dW_t. \label{eq:dyngX}
\end{align}
Integration by parts and the dynamics~\eqref{eq:dyngX} in turn yield
\begin{align}\label{eq:IBP}
d\left(\theta_t g(X_t)\right) = \big[ \dot{\theta_t}g(X_t) + \gamma \sigma^2 X_t\theta_t \big] dt 
+ \delta \theta_tg'(X_t)dW_t.
\end{align}
Here, the local martingale part is a true martingale. Indeed, by H\"older's inequality, the integrability condition~\eqref{eq:integrable} and the boundedness of $g'$ established in Lemma~\ref{symmetric},
\begin{align*}
\E\left[ \int_0^t |g'(X_u)|^2 \theta_u^2du \right] 
&\leq K^2\E\left[ \int_0^t \theta_u^{2}du \right]<\infty.
\end{align*}
Also taking into account that $G'(\dot{\varphi}^1_t) = G'(\left(G'\right)^{-1}(g(X_t))) = g(X_t)$, we can therefore use~\eqref{eq:IBP} to replace the second and the third terms on the right-hand side of~\eqref{eq:comp}, obtaining
\begin{align*}
J_T^1(\dot{\varphi}) - J_T^1(\dot{\varphi}^1) 
&\leq - \E[g(X_T)\theta_T] - \E\left[ \int_0^T \frac1 2 \gamma^1 \sigma^2\theta_t^2 dt \right].
\end{align*}
The Cauchy-Schwartz inequality yields
\begin{align*}
\big|\E[g(X_T)\theta_T]\big|
&\leq \sqrt{\E[g(X_T)^2]\E[\theta_T^2]} 
\leq\sqrt{\E[2g(X_T)^2]}\sqrt{\E[\varphi_T^2] + \E[(\varphi^1_T)^2]}.
\end{align*}
By the polynomial growth of $g$ established in Lemma~\ref{symmetric} and \eqref{sde:moments}, we have  $\sup_{T\geq 0} \E[g(X_T)^2] <\infty$. Together with the transversality condition~\eqref{assump:vanishing}, it follows that
\begin{align*}
0 \leq \lim_{T\to\infty} \frac{1}{T} \left|\E[g(X_T)\theta_T] \right|
\leq \lim_{T\to\infty} \frac{1}{T} \sqrt{\E[2g(X_T)^2]}\sqrt{ \E[\varphi_T^2] + \E[(\varphi^1_T)^2]}
= 0.
\end{align*}
Therefore, the trading rate $\dot{\varphi}^1$ is indeed long-run optimal for agent 1:
\begin{align*}
\limsup_{T\to\infty} \frac{1}{T} \left[J_T^1(\dot{\varphi}) - J_T^1(\dot{\varphi}^1)  \right] 
&\leq \limsup_{T\to\infty} \frac{1}{T} \left[ -\E[g(X_T)\theta_T] - \E\left[ \int_0^T \frac1 2 \gamma^1 \sigma^2\theta_t^2 dt \right]\right]\\
&= -\lim_{T\to\infty} \frac{1}{T} \E[g(X_T)\theta_T] + \limsup_{T\to\infty} \frac{1}{T}\E\left[- \int_0^T \frac1 2 \gamma^1 \sigma^2\theta_t^2 dt \right]\leq 0.
\end{align*}
An analogous argument shows that $\dot{\varphi}^2$ is long-run optimal for agent 2. This completes the proof.
\end{proof}

\subsection{Proofs for Section~\ref{sec:singular}}\label{sec:proofssingular}

The following lemma provides the counterpart of the function $g$ from Lemma~\ref{ODE} for proportional costs. It is given in closed form; its properties listed here are therefore easily verified by direct calculations:

\begin{lemma}\label{lem:gprop}
With the constants $l$ from~\eqref{eq:ell} and $\gamma$, $\delta$ from~\eqref{constants}, define
\begin{align}\label{ode:proportional}
g(x) = \frac{\gamma\sigma^2}{3\delta^2}\left(x^3 - 3l^2x\right) \one_{\{|x|\leq l\}}-\lambda\mathrm{sgn}(x) \one_{\{|x|>l\}}. 
\end{align}
This function has the following properties:
\begin{enumerate}
\item $g$ is an odd, decreasing function;
\item $\frac 1 2 \delta^2 g''(x) = \gamma \sigma^2 x$ for $x\in(-l, l)$;
\item $g'$ is continuous on $\R$ and $g'(l)=g'(-l)=0$;
\item For every $x\in[0,l]$, we have $0\geq g(x)\geq g(l)= -\lambda$.
\end{enumerate}
\end{lemma}

\begin{lemma}\label{1:strategy}
The strategies from Theorem~\ref{thm:singular} are admissible and satisfy the transversality condition~\eqref{singular:vanishing}. Moreover they clear the market.

\end{lemma}

\begin{proof}
Let $x = |\varphi_{0-}^1| +|\varphi_{0-}^2|+l+s$. 
First, note that the initial jump $X_0-X_{0-}$ satisfies
\begin{align*}
-l\leq X_0 = L_0-U_0 + X_{0-}\leq l,
\end{align*}
and hence
\begin{align*}
X_t = \delta W_t + L_t - U_t + X_{0-}. 
\end{align*}
Therefore, we have
\begin{align*}
\E[|L_T - U_T|] =\E[|X_T -  \delta W_T-X_{0-} |] \leq \delta \E[|W_T|] + \E[|X_T|]+|X_{0-}| \leq x + \delta \sqrt{\frac{2T}{\pi}}, 
\end{align*}
so that the transversality condition~\eqref{singular:vanishing} is satisfied. Next, notice that
\begin{align*}
|L_t - U_t|^2 \leq \left(|X_t| + |X_{0-}| + \delta |W_t| \right)^2 \leq \left(x +\delta|W_t|\right)^2\leq 2x^2 +2\delta^2|W_t|^2.
\end{align*}
As a consequence,
\begin{align*}
\E\left[\int_0^T(L_t - U_t)^2 dt \right] 
\leq \E\left[ \int_0^T 2x^2 +2\delta^2|W_t|^2 dt \right]   
= 2x^2T +2\delta^2\E\left[\int_0^T |W_t|^2 dt\right]
= 2x^2T + \delta^2T^2,
\end{align*}
so that $\varphi^1$ satisfies the first integrability condition in~\eqref{assump:singular}.

Now, apply It\^o's formula to $(X_T+l)^2/4l$, obtaining
\begin{align*} 
&\frac{1}{4l}(X_T+l)^2 - \frac{1}{4l}(X_0+l)^2 \\
&\quad= \int_0^T \frac{\delta}{2l}(X_t+l) dW_t + \int_0^T \frac{\delta^2}{4l} dt + \int_0^T \frac{1}{2l}(-l+l) dL_t - \int_0^T \frac{1}{2l}(l+l) dU_t\nonumber\\
&\quad= \int_0^T \frac{\delta}{2l}(X_t+l) dW_t + \frac{\delta^2}{4l} T - U_T + U_0.
\end{align*}
Rearranging, taking expectations, and taking into account that $0\leq U_0\leq |X_0| \leq x$ leads to
\begin{align}\label{rulesu}
\E[U_T] 
&= U_0 +\frac{1}{4l}(X_0+l)^2 + \frac{\delta^2}{4l} T - \E\left[ \int_0^T \frac{\delta}{2l}(X_t+l) dW_t \right] - \E\left[\frac{1}{4l}(X_T+l)^2\right] \nonumber\\
&\leq x+l + \frac{\delta^2}{4l} T. 
\end{align}
After applying It\^o's formula to $(X_T-l)^2/4l$, a symmetric calculation and $0\leq L_0\leq |X_0| \leq x$ show
\begin{align}\label{rulesl}
\E[L_T] 
&= L_0 +\frac{1}{4l}(X_0-l)^2 + \frac{\delta^2}{4l} T - \E\left[ \int_0^T \frac{\delta}{2l}(X_t-l) dW_t \right] - \E\left[\frac{1}{4l}(X_T-l)^2\right] \nonumber\\
&\leq x+l + \frac{\delta^2}{4l} T. 
\end{align}
Combining \eqref{rulesu} and \eqref{rulesl} yields the second integrability condition in~\eqref{assump:singular}; therefore $\varphi^1$ is indeed admissible. Market clearing evidently holds by construction; in particular $\varphi^2$ is admissible as well. For later use also observe that, by definition,
\begin{align}\label{eq:repprop}
\varphi^1_t =X_t - \delta W_t +\frac{s\gamma^2}{\gamma^1 + \gamma^2},\quad \quad \gamma^1\sigma( \sigma\varphi^1_t  + \beta^1_t) - \mu_t = \gamma\sigma^2 X_t.
\end{align}
\end{proof}

\begin{proof}[Proof of Theorem~\ref{thm:singular}]
Consider a competing admissible strategy with Jordan-Hahn decomposition $\varphi=\varphi^1_{0-}+\varphi^{\uparrow}-\varphi^{\downarrow}$. To ease notation, set 
$$
\theta_t = \varphi_t - \varphi^1_t, \qquad \mbox{so that} \quad d\theta_t = d\varphi^{\uparrow}_t - d\varphi^{\downarrow}_t - dL_t + dU_t, \qquad \theta_{0-} = 0.
$$
By properties (i) and (iv) of the function $g$ from Lemma~\ref{lem:gprop}, we have 
\begin{align}
\one_{(-l,0)}(X_t)g(X_t)d\theta_t &
\leq \lambda\one_{(-l,0)}(X_t)\left[d\varphi^{\uparrow}_t + d\varphi^{\downarrow}_t + dU_t\right], \label{i1}\\
\one_{(0,l)}(X_t)g(X_t)d\theta_t &
\leq \lambda\one_{(0,l)}(X_t)\left[d\varphi^{\uparrow}_t + d\varphi^{\downarrow}_t + dL_t\right].\label{i2}
\end{align}
Since $L$, $U$ only grow on the sets $\{X_t=-l\}$ and $\{X_t=l\}$, respectively, properties (i) and (iv) of $g$ from Lemma~\ref{lem:gprop} and (\ref{i1}-\ref{i2}) show that
\begin{align*}
&\int_{0-}^T g(X_t)d\theta_t \nonumber\\
& = \lambda\int_{0-}^T \one_{\{-l\}}(X_t) \left[ d\varphi^{\uparrow}_t - d\varphi^{\downarrow}_t - dL_t\right] 
- \one_{\{l\}}(X_t) \left[ d\varphi^{\uparrow}_t - d\varphi^{\downarrow}_t + dU_t\right]
+ \one_{(-l,l)}(X_t)g(X_t)d\theta_t \nonumber\\
& \leq \lambda\int_{0-}^T \one_{\{-l\}}(X_t) \left[ d\varphi^{\uparrow}_t - d\varphi^{\downarrow}_t - dL_t\right] 
+ \one_{\{l\}}(X_t) \left[ d\varphi^{\downarrow}_t - d\varphi^{\uparrow}_t - dU_t\right]
+ \one_{(-l,l)}(X_t)\left[ d\varphi^{\uparrow}_t + d\varphi^{\downarrow}_t \right]\nonumber\\
& \leq\lambda \int_{0-}^T \one_{\{-l\}}(X_t) \left[ d\varphi^{\uparrow}_t +d\varphi^{\downarrow}_t - dL_t\right] 
+ \one_{\{l\}}(X_t) \left[ d\varphi^{\uparrow}_t + d\varphi^{\downarrow}_t - dU_t\right]
+ \one_{(-l,l)}(X_t)\left[ d\varphi^{\uparrow}_t + d\varphi^{\downarrow}_t \right]\nonumber\\
& \leq\lambda\int_{0-}^T \left(\one_{\{-l\}}(X_t) + \one_{(-l,l)}(X_t) + \one_{\{l\}}(X_t)\right) \left[ d\varphi^{\uparrow}_t + d\varphi^{\downarrow}_t - dL_t -dU_t\right] \nonumber\\
& = \lambda \left[\varphi^{\uparrow}_T + \varphi^{\downarrow}_T - L_T - U_T\right] - \lambda \left[\varphi^{\uparrow}_{0-} + \varphi^{\downarrow}_{0-} - L_{0-} - U_{0-}\right]\\
& = \lambda \left[\varphi^{\uparrow}_T + \varphi^{\downarrow}_T - L_T - U_T\right] . 
\end{align*}
Together with~\eqref{eq:repprop}, it follows that
\begin{align}
&J^1_T({\varphi}) - J^1_T(\varphi^1) \notag\\
&= \E\left[ \int_{0-}^T\left(\big(\varphi_t - \varphi^1_t \big)\mu_t - 
\frac{\gamma^1}{2} \sigma\big(\varphi_t - \varphi^1_t \big)(\sigma\varphi_t+ \sigma\varphi^1_t +2\beta^1_t)\right) dt
-\lambda(\varphi^{\uparrow}_T  + \varphi^{\downarrow}_T) + \lambda(L_T +U_T) \right] \notag\\
& = \E\left[ \int_{0-}^T\left(\big(\varphi_t - \varphi^1_t \big)\mu_t - 
\frac{\gamma^1}{2} \sigma\big(\varphi_t - \varphi^1_t \big)\big(\sigma\varphi_t - \sigma\varphi^1_t +2(\sigma\varphi^1_t + \beta^1_t)\big) \right) dt
- \lambda \big(\varphi^{\uparrow}_T + \varphi^{\downarrow}_T - L_T - U_T\big) \right] \notag\\
&= \E\left[ \int_{0-}^T-\left(\frac1 2 \gamma^1 \sigma^2\big(\varphi_t - \varphi^1_t \big)^2 + \gamma \sigma^2 X_t\big(\varphi_t - \varphi^1_t \big) \right) dt -
 \lambda \big(\varphi^{\uparrow}_T + \varphi^{\downarrow}_T - L_T - U_T\big) \right]\notag\\
&\leq -\E\left[ \int_{0-}^T \frac1 2 \gamma^1 \sigma^2\theta_t^2 dt \right] -\E\left[\int_{0-}^T \gamma \sigma^2 X_t\theta_t dt + \int_{0-}^T g(X_t)d\theta_t \right] 
. \label{eq:compprop}
\end{align}
To simplify this expression, use It\^o's formula, the dynamics~\eqref{reflected BM} of the doubly-reflected Brownian motion $X$, the fact that $L$, $U$ only grow on the sets $\{X_t=-l\}$ and $\{X_t=l\}$ respectively, and the ODE for $g$ from Lemma~\ref{lem:gprop}(ii) to compute
\begin{align*}
dg(X_t) 
&= \frac 1 2 \delta^2 g''(X_t) dt + g'(X_t) \big[dL_t - dU_t \big] + \delta g'(X_t)
dW_t\\
&= \gamma\sigma^2 X_t dt +  \delta g'(X_t)
dW_t.
\end{align*}
Integration by parts in turn yields
\begin{align*}
d\left(g(X_t)\theta_t\right) = g(X_t)d\theta_t + \gamma\sigma^2 \theta_t X_t dt + \delta \theta_tg'(X_t)
dW_t. 
\end{align*}
Since $g'$ is bounded, the integrability condition~\eqref{assump:singular}  implies that the local martingale part in this decomposition is a true martingale, so that
\begin{align}\label{eq:local}
\E\left[ \int_{0-}^T \gamma \sigma^2 X_t\theta_t dt + \int_{0-}^T g(X_t)d\theta_t \right] =  \E\left[g(X_T)\theta_T\right] - \E\left[g(X_{0-})\theta_{0-}\right] = \E\left[g(X_T)\theta_T\right].
\end{align}
Now, the long-run optimality of $\varphi^1$ for agent $1$ follows from~\eqref{eq:compprop} and~\eqref{eq:local} by taking into account that property (iv) of $g$ and the transversality condition~\eqref{singular:vanishing} imply
\begin{align*}
\lim_{T\to\infty}\frac{1}{T} \big|\E\left[g(X_T)\theta_T\right]\big|
&\leq \lim_{T\to\infty}\frac{1}{T} \E\left[|g(X_T)\theta_T|\right]
\leq \lim_{T\to\infty}\frac{\lambda}{T} \E\left[|\theta_T|\right]
\leq \lim_{T\to\infty}\frac{\lambda}{T} \E\left[|\varphi_T| + |\varphi^1_T|\right] = 0.
\end{align*}
An analogous argument shows that $\varphi^2$ is optimal for agent 2, thereby completing the proof.
\end{proof}

\appendix

\section{Proof of Lemma~\ref{ODE}}\label{app:ode}
In this appendix, we establish existence, uniqueness, and properties for the second-order nonlinear ODE~\eqref{eqn:ergodic ODE} from Lemma~\ref{ODE}. To this end, we introduce the following first-order nonlinear ODE:
\begin{equation}\label{ode:general}
y'(x) = f(x,y(x)) = -ax^2 + b + F(y(x)),
\end{equation}
and extend the ideas of~\cite{guasoni.weber.18} to general functions $F:\R\to\R$ which satisfy Assumption~\ref{property} below. That is, in Lemma~\ref{general ODE}, 
 we establish that for suitable functions $F$, and any choice of $a>0$ and $b\in\R$, \eqref{ode:general} has a unique positive solution on its maximal domain which contains $[\sqrt{\max\{b,0\}/a},\infty)$. 
Then, for the first-order ODE
\begin{align}\label{ODE:g}
g'(x) = ax^2 - b - F(g(x)),
\end{align}
Lemma~\ref{symmetric} shows that there is a unique value of $b$ that guarantees there is a solution on $\R$ such that $xg(x)\leq 0$, and the solution is unique. Moreover, Lemma~\ref{2nd ODE} proves that this solution to~\eqref{ODE:g} is also the unique solution of the second-order ODE
\begin{equation}\label{2nd ODE:g}
g''(x) = 2ax - F'(g(x))g'(x).
\end{equation}
Finally, with the help of Lemma~\ref{equ assump} pointing out the relationship between Assumption~\ref{cond:cost} and Assumption~\ref{property}, we establish the proof of Lemma~\ref{ODE} with $F$ chosen to be proportional to the Legendre transform of the trading cost function $G$. 

To carry out this program, we first introduce the assumptions on $F$ that are needed to generalize the argument developed for power functions by \cite{guasoni.weber.18}. Subsequently, in Remark~\ref{prop1} and Lemma~\ref{prop2}, we derive a number of consequences, which are crucial tools for the analysis.  

\begin{assumption}\label{property}
\begin{enumerate}[(i)]
\item $F$ is convex, differentiable, even, and strictly increasing on $[0,\infty)$ with $F(0)=0$;
\item $F'$ is also differentiable and strictly increasing on $[0,\infty)$ with $F'(0)=0$; 
\item There exists a constant $K$ such that $F(x)\leq K(1+|x|^p)$ for some $p\geq2$;
\item There exist constants $\tilde{C}>0$ and $x_0>0$ such that $F''(x)>\tilde{C}$ for every $|x|>x_0$. 
\end{enumerate}
\end{assumption}

\begin{remark}\label{prop1}
Some immediate consequences of Assumption~\ref{property} are as follows:
\begin{enumerate}[(i)]
\item $F'$ is increasing on the whole real line, since it is an odd function (as $F$ is even) and $F'$ is strictly increasing on $[0,\infty)$;
\item Assumption (iv) implies that there is some $\hat{a}>0$ such that $F(x)>\hat a x^2$ for large $x>0$. This is why $p\geq 2$ in Assumption~\ref{property}(iii) is without loss of generality.
\end{enumerate}
\end{remark}

\begin{lemma}\label{prop2}
Suppose $F$ satisfies Assumption~\ref{property}. Then:
\begin{enumerate}[(i)]
\item $F^{-1}$ exists and is concave on $[0,\infty)$;
\item For every $x\geq0$ and every $\alpha\geq 1$:
\begin{align*}
\alpha F(x) \leq F(\alpha x), \quad \quad F^{-1}(\alpha x) \leq \alpha F^{-1}(x);
\end{align*}
\item For $x,y\geq 0$:
\begin{align*}
F(x+y) \geq F(x)+F(y), \qquad F^{-1}(x) +F^{-1}(y)\geq F^{-1}(x+y);
\end{align*}
\item On $(0,\infty)$, $F^{-1}$ is strictly increasing but $(F^{-1})'$ is strictly decreasing; 
\item There exists constant $C>0$ that ${F^{-1}}(x^2) \leq C|x|$ and $2x(F^{-1})'(x^2)\leq 2C$ for every $|x|>x_0$. 
\end{enumerate}
\end{lemma}

\begin{proof} (i): Convexity of $F$ implies that, for $x,y\geq 0$ and $0<a<1$, 
\begin{align*}
ax + (1-a)y = aF(F^{-1}(x))+(1-a)F(F^{-1}(y))\geq F(aF^{-1}(x)+(1-a)F^{-1}(y)).
\end{align*}
As $F$ is increasing, $F^{-1}$ is increasing as well. Applying $F^{-1}$ on both sides of the above estimate in turn yields the concavity of $F^{-1}$.  

\hspace{6mm}(ii): Recall that $F(0)=0$ and again use convexity of $F$ to obtain, for every $x\geq0$ and $\alpha\geq1$,
$$ 
F(\alpha x)= \alpha\left[\frac{1}{\alpha} F(\alpha x) + \left(1-\frac{1}{\alpha}\right) F(0) \right]\geq \alpha F\left(\frac{1}{\alpha} \alpha x\right) = \alpha F(x).
$$
Analogously, the concavity of $F^{-1}$ yields $F^{-1}(\alpha x) \leq \alpha F^{-1}(x)$.

\hspace{6mm}(iii): Since $F'$ is increasing we have $F'(x+y) - F'(x) \geq0$ for every $x,y>0$. As a consequence, $F(x+y) - F(x) \geq F(0+y) - F(0) = F(y)$ as asserted. 
It implies that 
$$
F(F^{-1}(x) + F^{-1}(y)) \geq F(F^{-1}(x) ) +F(F^{-1}(y))  = x + y.
$$
Since $F$ is strictly increasing, we can infer 
$$
F^{-1}(x) + F^{-1}(y) \geq F^{-1}(x+y). 
$$

\hspace{6mm}(iv): Since $F$ is convex and $F$ and $F'$ are strictly increasing on $[0,\infty)$, then $F'\geq 0$, $F''\geq0$ and they are both not equal to zero on any interval, hence 
\begin{align*}
(F^{-1})' (x)= \frac{1}{F'(F^{-1}(x))}\geq 0, \qquad (F^{-1})'' (x)= -\frac{F''(F^{-1}(x))}{\Big(F'(F^{-1}(x))\Big)^3}\leq 0,
\end{align*}
and they are both not zero on any interval. 
So $F^{-1}$ is strictly increasing on $[0,\infty)$ but
$(F^{-1})'$ is strictly decreasing on $(0,\infty)$ as asserted.
 
\hspace{6mm}(v): By directly integrating the inequality in Assumption~\ref{property} (iv) and choosing $C$ large enough, together with (ii) in~\ref{prop2}, it's easy to see that the first statement holds. For the second statement, by Assumption~\ref{property}(ii),
\begin{align*}
\frac{d}{dx}\left[ x F'(x) - F(x)\right] = xF''(x) + F'(x) - F'(x) = xF''(x)\geq0.
\end{align*}
As $F^{-1}$ is strictly increasing, we have $F^{-1}(x^2) >0$ for $x>0$ and hence
\begin{align*}
F'(F^{-1}(x^2)) \geq \frac{F(F^{-1}(x^2))}{F^{-1}(x^2)} = \frac{x^2}{F^{-1}(x^2)}. 
\end{align*}
Together Assumption~\ref{property} (iv) it follows that, for $x\geq|x_0|$, 
\begin{align*}
\frac{d}{dx} F^{-1}(x^2) = 2x (F^{-1})'(x^2) = \frac{2x}{F'(F^{-1}(x^2))} \leq \frac{2xF^{-1}(x^2)}{x^2} \leq \frac{2Cx^2}{x^2} = 2C, 
\end{align*} 
which yields the desired result.
\end{proof}

Now we address the existence and uniqueness of the positive solution to~\eqref{ode:general} on $[\sqrt{b/a},\infty)$. 

\begin{lemma}\label{general ODE}
Let $F$ be a function satisfying Assumption~\ref{property} and $a>0$, $b\in\R$. Then there exists a unique solution $y$ of 
\begin{equation}
y'(x) = f(x,y(x)) = -ax^2 + b + F(y(x)), \tag{\ref{ode:general}}
\end{equation}
such that $[\sqrt{\max\{b,0\}/a}, \infty)$ is contained in its maximal interval of existence, and $y(x)\geq 0$ for every $x\geq \sqrt{\max\{b,0\}/{a}}$. Moreover, $[0, \infty)$ is contained in its maximal interval of existence, $y$ is increasing on $[\sqrt{\max\{b,0\}/{a}}, \infty)$, and satisfies the growth condition
\begin{equation}\label{lemma_growth}
\lim_{x\to \infty} \frac{y(x)}{F^{-1}(ax^2)} =1. 
\end{equation}
Further, in \eqref{ode:general}, if we write $y(x) = y(x;b)$, then for every $x\in [0,\infty)$, $b\in\R$,
\begin{align}\label{eq:b}
\frac{\partial{y(x;b)}}{\partial b} < 0.
\end{align}
\end{lemma}

\begin{proof} 
Let ${b}_+ = \max\{b,0\}$. 
On $[\sqrt{{b}_+/a},+\infty)$, define the function $h(x)=F^{-1}({a x^2-b})$. 
Notice that by definition of $h(x)$ we have $f(x,h(x))=0$ and $h$ is strictly increasing on $[\sqrt{b_+/a},+\infty)$. Thus we can infer that $h$ is a supersolution on $(\sqrt{b_+/a}, \infty)$ in that $h'(x) \geq f(x,h(x))=0$.

Notice that $f(x,y)$ is locally Lipschitz, so that local existence and uniqueness hold for the initial-value problem~\eqref{ode:general} with initial condition $(x_0,y_0)$. For every $\bar x>\sqrt{b_+/a}$, let $y(x;\bar x, h(\bar x))$ denote the unique solution to~\eqref{ode:general} with initial condition $(\bar x, h(\bar x))$ on its maximal interval of existence $(T^{-}, T^{+})$. 
The first step is to show that the following inequalities hold for every $\bar{x}>\sqrt{b_+/a}$:
\begin{align}
&\mbox{on } [\sqrt{b_+/a},\bar{x}),\;  y({x};\bar x, h(\bar x)) > h(x)\geq 0, \; y(\cdot;\bar x, h(\bar x)) \mbox{ is increasing},  \label{est: x<barx} \\
&\mbox{on } [\bar{x},T^+), \;\;\quad y({x};\bar x, h(\bar x)) \leq h(x).\label{est: x>barx}
\end{align}
First, by directly calculating the first-order derivative, we find that for every $\bar{x}>\sqrt{b_+/a}$, 
\begin{align*}
y'(\bar{x};\bar x, h(\bar x)) = f(\bar{x},y(\bar{x};\bar x, h(\bar x)) )
=f(\bar{x}, h(\bar x))= 0 < h'(\bar{x}).
\end{align*}
Therefore, there exists $\epsilon^{-}\in(0,\bar{x} - \sqrt{b_+/a} \vee T^{-})$ such that $y({x};\bar x, h(\bar x)) >h(x)$ for $x\in(\bar{x}-\epsilon^{-},\bar{x})$. Define
\begin{align*}
x_0 = \inf\{x\in[\sqrt{b_+/a},\bar{x})\cap(T^{-}, T^+): y({x};\bar x, h(\bar x)) >h(x)\}.
\end{align*}
It is easy to see that on $(x_0,\bar{x})$, $y({x};\bar x, h(\bar x))$ is increasing through
\begin{align*}
y'({x};\bar x, h(\bar x)) = f({x},y({x};\bar x, h(\bar x)) ) > f(x,h(x))=0,
\end{align*}
and since $y(x_0;\bar x, h(\bar x))$ is between $h(x_0)>0$ and $h(\bar{x})<\infty$, we conclude $x_0\in(T^{-},T^+)$.  
Suppose $x_0>\sqrt{b_+/a}$. The definition of $x_0$ yields $y({x_0};\bar x, h(\bar x)) = h(x_0)$ and $y'({x_0};\bar x, h(\bar x)) \geq h'(x_0)>0$; but plugging $y({x_0};\bar x, h(\bar x)) = h(x_0)$ into~\eqref{ode:general} gives $y'({x_0};\bar x, h(\bar x)) = f(x_0,h(x_0)) = 0$, a contradiction. Therefore, $y(\cdot;\bar{x},h(\bar{x}))$ is increasing on $(x_0,\bar{x})$ and $x_0 = \sqrt{b_+/a}$, which implies that~\eqref{est: x<barx} holds. 

To show~\eqref{est: x>barx}, we calculate the second-order derivative, 
\begin{align}\label{eq:2nd derivative}
y''(x;\bar x, h(\bar x)) = -2ax + F'(y(x;\bar x, h(\bar x))) y'(x;\bar x, h(\bar x)),
\end{align}
which implies that $y''(\bar{x};\bar x, h(\bar x))< 0$ and there exists $\epsilon^{+} >0$ such that $y(x;\bar x, h(\bar x))<h(x)$ for $x\in(\bar{x}, \bar{x}+\epsilon^{+})$. 
Define
\begin{align*}
x_1 = \sup\{x\in[\bar{x},T^{+}): y({x};\bar x, h(\bar x)) <h(x)\}.
\end{align*}
Suppose $x_1<T^{+}$; the definition of $x_1$ yields $y({x_1};\bar x, h(\bar x)) = h(x_1)$ and $y'({x_1};\bar x, h(\bar x)) \geq h'(x_1) >0$. But plugging $y({x_1};\bar x, h(\bar x)) = h(x_1)$ into~\eqref{ode:general} gives $y'({x_1};\bar x, h(\bar x)) = f(x_1,h(x_1)) = 0$, a contradiction. Therefore, \eqref{est: x>barx} holds.

Now define
\begin{align}
x_2: = \sup\{x\geq \bar{x}: y(x;\bar{x},h(\bar{x}))\geq 0, \text{ or } y'(x;\bar{x},h(\bar{x}))\leq 0\}.
\end{align}
From~\eqref{eq:2nd derivative}, we see that $y''(x;\bar{x},h(\bar{x})) \leq -2ax$ for $x\in[\bar{x}, x_2)$, so $y(\cdot;\bar{x},h(\bar{x}))$ is strictly decreasing and strictly concave on $[\bar{x},x_2)$. This implies $x_2<+\infty$, and by continuity we have $y(x_2;\bar{x},h(\bar{x})) = 0$, $y'(x_2;\bar{x},h(\bar{x})) <0$, and in addition that $y(\cdot;\bar{x},h(\bar{x}))<0$ in a right-neighbourhood of $x_2$. For $x>x_2$, we claim that $y(x;\bar{x},h(\bar{x})) \leq 0$, because in order to become positive again, $y(\cdot;\bar{x},h(\bar{x}))$ would need to cross zero, but $y(x;\bar{x},h(\bar{x}))=0$ implies $y'(x;\bar{x},h(\bar{x}))=-ax^2+b  <0$. Therefore, we can conclude that either $T^+ = \infty$ or $\lim_{x\uparrow T^+} y(x;\bar x, h(\bar x)) = -\infty$.
For the latter case, define
\begin{align*}
y(x;\bar x, h(\bar x)) = -\infty, \qquad \mbox{for } x\in[T^+,\infty).
\end{align*}

Next, we consider the relationship between $y(x; \bar x_1, h(\bar x_1))$ and $y(x; \bar x_2, h(\bar x_2))$ for $\bar x_2 > \bar x_1>\sqrt{b_+/a}$. By~\eqref{est: x<barx}, at $\bar x_1$, 
$y(\bar x_1; \bar x_1, h(\bar x_1)) = h(\bar x_1) < y(\bar x_1; \bar x_2, h(\bar x_2))$. By (local) uniqueness of the initial value problems associated with~\eqref{ode:general}, there cannot exist $x$ such that $y(x; \bar x_1, h(\bar x_1))=y(x; \bar x_2, h(\bar x_2)) >-\infty$, thus the graph of $y(x; \bar x_1, h(\bar x_1))$ lies strictly below the graph of $y(x; \bar x_2, h(\bar x_2))$ except when they both take the value $-\infty$. In summary
\begin{align}\label{increasing in bar x}
\mbox{on } (\sqrt{b_+/a}, \infty), \quad  y(x; \cdot, h(\cdot)) \mbox{ is increasing}. 
\end{align}

Next, we show that any solution $y$ of~\eqref{ode:general} such that $[0, \infty)$ is contained in its maximum interval of existence with $y(x)\geq 0$ for every $x\geq \sqrt{{b_+}/{a}}$, automatically satisfies the growth condition~\eqref{lemma_growth}. From the above argument concerning the relationship between $h(x)$ and $y(x; \bar{x}, h(\bar{x}))$, an important observation is for every $x>\sqrt{b_+/a}$ and every $\bar{x}>x$, we need to have $y(x)>y(x;\bar{x},h(\bar{x}))\geq h(x)$; otherwise the solution $y$ will not stay positive. We summarize the properties of $y$ as follows: 
\begin{enumerate}[i)]
\item\label{positive} $y(x)> h(x)\geq 0$, ${y}'(x) = -ax^2+b+F(y(x)) > -ax^2+b+F(h(x))=0$, which means $y$ is strictly increasing on $(\sqrt{b_+/a},+\infty)$;
\item\label{domain} $[\sqrt{b_+/a},+\infty)\subset D$, where $D$ is the maximal interval of existence of $y(x)$. 
\end{enumerate}
From Property~\ref{positive} and Lemma~\ref{prop2} (iii,iv), it follows that
\begin{align*}
1 &= \lim_{x\to\infty} \frac{F^{-1}(ax^2) - F^{-1}(b_+)}{F^{-1}(ax^2)}
\leq\liminf_{x\to\infty}\frac{h(x)}{F^{-1}(ax^2)} \leq \limsup_{x\to\infty}\frac{h(x)}{F^{-1}(ax^2)} 
\leq \lim_{x\to\infty} \frac{F^{-1}(ax^2)}{F^{-1}(ax^2)}=1,
\end{align*}
and in turn
\begin{equation*}
\liminf_{x \rightarrow \infty} \frac{y(x)}{F^{-1}(a x^2)}\geq 1.    
\end{equation*}
Next we show that $L=\lim_{x \rightarrow \infty} \frac{y(x)}{F^{-1}(a x^2)}$ exists and $L=1$. To this end, set $M=\limsup_{x \rightarrow \infty} \frac{y(x)}{F^{-1}(a x^2)}$ and notice that $1\leq M\leq\infty$. If $M=1$ then we can conclude that $L=1$. 

Assume $1<M<\infty$. We first want to show $M=L$. There exists a  sequence $(x_n)_{n\geq0} \to \infty$ such that 
\begin{align*}
\lim_{n \rightarrow \infty} \frac{y(x_n)}{F^{-1}(a x_n^2)}=M.
\end{align*} 
In particular, for any $\delta\in(0,M-1)$ there exists $N_\delta \in \mathbb{N}$ such that for every $n \geq N_\delta$ we have 
\begin{align*}
y(x_n)\geq(M-\delta){F^{-1}(a x_n^2)}.
\end{align*}
For large $x$, we claim that the function $s(x)=(M-\delta)F^{-1}(a x^2)$ is a subsolution of \eqref{ode:general}. By Lemma~\ref{prop2} (v), we know that for $x\geq|x_0|/\sqrt{a}$, 
\begin{align*}
0<s'(x) = (M-\delta){(F^{-1})'(a x^2 )} 2ax \leq 4\sqrt{a}(M-\delta)C.
\end{align*}
Since $M-\delta>1$, there exists $\bar{x}$ such that for $x\geq\bar{x}$, we have $(M-\delta)ax^2 -ax^2 +b \geq 4\sqrt{a}(M-\delta)C$. As a consequence, 
\begin{align*}
s'(x)\leq 4\sqrt{a}(M-\delta)C &\leq-ax^2+b+(M-\delta)ax^2 \\
&=-ax^2+b+F(F^{-1}((M-\delta)(ax^2)))\\
&\leq-ax^2+b+F((M-\delta)F^{-1}(ax^2)) = f(x,s(x)).
\end{align*}
On the other hand, notice that $y(x_n)\geq s(x_n)$ for every $n\geq N_\delta$. 
Thus by the comparison lemma, for every $\delta\in(0,M-1)$ and some large $x_N$, from $y(x_N) \geq (M-\delta) F^{-1}(a x_N^2) = s(x_N)$ we can conclude $y(x)\geq s(x) =(M-\delta){F^{-1}(a x^2)}$ for $x\geq x_N$. 
In particular, for every small $\delta$,
$$
\liminf_{x \rightarrow \infty} \frac{y(x)}{F^{-1}(a x^2)}\geq M-\delta,
$$
and therefore
$$
\liminf_{x \rightarrow \infty} \frac{y(x)}{F^{-1}(a x^2)}=M=\limsup_{x \rightarrow \infty} \frac{y(x)}{F^{-1}(a x^2)}.
$$
 If $M=\infty$, we substitute $M-\delta$ with $N\in\mathbb{N}$ and then infer with the same argument that $\liminf_{x \rightarrow \infty} \frac{y(x)}{F^{-1}(a x^2)}=\infty$.  In other words, the limit $L$ exists and $L=M \in[1, \infty]$.
 
Next, we show $L=1$. First, assume to the contrary $1<L<\infty$. 
Since $\lim_{x \rightarrow \infty} \frac{y(x)}{F^{-1}(a x^2)}=L$, by Lemma~\ref{prop2} (v), there exists a constant $K>0$ such that $y(x)\leq Kx$ for large $x>0$. Moreover, for every $\delta\in(0,L-1)$ and large $x$, by Lemma~\ref{prop2} (ii),
\begin{align*}
y(x) \geq (L-\delta)F^{-1}(ax^2)\geq F^{-1}((L-\delta)ax^2).
\end{align*}
As a consequence,
\begin{align*}
\liminf_{x\to\infty} \frac{F(y(x))}{ax^2} \geq L-\delta.
\end{align*}
On the other hand, \eqref{ode:general} implies 
\begin{equation*}
\liminf_{x\rightarrow\infty}\frac{y'(x)}{a x^2}\geq L-\delta-1>0,
\end{equation*}
so that $y'(x)$ grows at least quadratically, leading to a contradiction. 

Now assume that $L=+\infty$. With a similar argument as above, for every $L'>0$ and $x$ sufficiently large,
\begin{align*}
F(y(x)) \geq F(F^{-1}(L'ax^2)) = L'ax^2, 
\end{align*}
and it follows that 
\begin{align*}
\lim_{x\to\infty} \frac{ax^2}{F(y(x))} = 0.
\end{align*}
From~\eqref{ode:general} it follows that
\begin{equation*}
\lim_{x\rightarrow\infty}\frac{y'(x)}{F( y(x))} = \lim_{x\to\infty} \frac{-ax^2+b+F(y(x))}{F(y(x))} =1.
\end{equation*}
Notice that for large $x$ such that $y(x)>C x_0$, Lemma~\ref{prop2} (v) yields
$$
\frac{y(x)^2}{C^2} = F\left(F^{-1}\left(\frac{y(x)^2}{C^2}\right)\right) \leq F\left(C \frac{y(x)}{C}\right) = F(y(x)).
$$
Thus, for small $\delta$ and sufficiently large $\bar{x}$, for all $x>\bar{x}$ we have
\begin{align*}
 \frac{1-\delta}{C^2} y^2(x)\leq (1-\delta)F(y(x))\leq y'(x).
 \end{align*} 
Hence for sufficiently large $\xi>\bar{x}$,
\begin{align*} 
\frac{y'(\xi)}{y^2(\xi)} \geq  \frac{1-\delta}{C^2}  >0.
\end{align*}
Integrating this inequality from $\xi = \bar{x}$ to $\xi=x$, we obtain
\begin{align*}
y(x) \geq \frac{1}{\frac{1}{y(\bar x)}-  \frac{1-\delta}{C^2}(x-\bar x)}.
\end{align*}
In particular, $y(x)$ has a vertical asymptote, contradicting Property~\ref{domain}. In summary, $L=1$. 

We now establish the uniqueness of $y(x)$. Suppose there exists another solution $y_2$ of~\eqref{ode:general} such that $[0, \infty)$ is contained in its maximal domain, and $y_2(x)\geq 0$ for every $x\geq \sqrt{{b_+}/{a}}$, and there exists $\bar x$, $\delta>0$ such that $y_2(\bar x) \geq y(\bar x) +\delta $. Then, on $[\bar x, \infty)$, the graph of $y_2$ always lies above $y$; otherwise it will violate the local uniqueness of the initial value problems associated with~\eqref{ode:general}. Moreover, for $x\geq\bar{x}$,
\begin{align*}
y_2'(x) - y'(x) = F(y_2(x)) - F(y(x))\geq 0,
\end{align*}
which means $y_2 - y$ is increasing. As a result,
\begin{align*}
y_2'(x) - y'(x) &= F(y_2(x)) - F(y(x)) 
\geq F(y_2(x) - y(x)) 
\geq F(y_2(\bar x) - y(\bar x)) 
\geq F(\delta)>0,
\end{align*}
which implies that, for every $x>\bar x$,
\begin{align*}
y_2(x) - y(x) \geq \delta + (x-\bar x) F(\delta).
\end{align*}
But $y_2$ also satisfies~\eqref{lemma_growth}, and for large $x$ we have 
\begin{align*}
F^{-1}(ax^2)\leq \sqrt{a}C x.
\end{align*}
Whence,
\begin{align*}
0 = \lim_{x\to\infty} \frac{y_2(x) - y(x)}{F^{-1}(ax^2)}  \geq \lim_{x\to\infty} \frac{\delta + (x-\bar x) F(\delta)}{\sqrt{a}Cx} = \frac{F(\delta)}{\sqrt{a}C}>0,
\end{align*}
which leads to contradiction. A symmetric argument yields the same results for the case where there exists $\bar x$ and $\delta>0$ such that $y_2(\bar x) \leq y(\bar x)-\delta$. This establishes uniqueness. 

We now establish the existence of $y(x)$. To this end, fix $x\geq \sqrt{b_+/a}$ and define 
\begin{align*}
y_*(x)=\sup\{y(x;\bar x,h(\bar x)): \bar x > x\}. 
\end{align*} 
Let $x_0>0$ and $C>0$ be the constant in Lemma~\ref{prop2} (v). For every $x_1\geq\sqrt{b_+/a}$, we can choose a large 
$y_1 > F^{-1} (ax_1^2 +2\sqrt{a}C+x^2_0+|b|)$, 
and for $x\geq x_1$ define
\begin{align*}
\tilde y(x)=F^{-1}(F(y_1)+{a}(x^2-x_1^2)).
\end{align*}
Then by $F(y_1)-ax_1^2 +b> 0$, $\tilde y(x)>h(x)$ for every $x\geq x_1$. Moreover, from the fact that $F(y_1)-ax_1^2 +b > 2\sqrt{a}C +x^2_0 + |b|+b>x^2_0$, we can infer
\begin{align*}
0\leq {\tilde y}'(x) &= 2ax (F^{-1})'(F(y_1)+{a}(x^2-x_1^2)) 
\\& \leq 2\sqrt{a} \sqrt{F(y_1)+{a}(x^2-x_1^2)}  (F^{-1})'(F(y_1)+{a}(x^2-x_1^2))
\\&\leq 2\sqrt{a}C 
\\&< F(y_1) - ax_1^2 +b = f(x_1,{\tilde y}(x_1)).
\end{align*} 
In particular, the unique local solution $y(x;x_1,y_1)$ to \eqref{ode:general} with initial condition $(x_1,y_1)$ satisfies 
\begin{align*}
{\tilde y}'(x_1) < f(x_1,{\tilde y}(x_1)) = f(x_1, y_1) = y'(x_1;x_1,y_1).
\end{align*}
Thus for every $\bar x > x_1$, $y(\bar x; x_1,y_1)>\tilde y(\bar x) > h(\bar x) = y(\bar{x}; \bar{x}, h(\bar x))$. The local uniqueness of the initial value problems associated with~\eqref{ode:general} implies that $y(x; x_1, y_1)$ and $ y(\bar{x}; \bar{x}, h(\bar x))$ cannot cross, so the graph of $y(x;x_1,y_1)$ lies above $y(x;\bar x, h(\bar x))$ and, in particular, $y_1=y(x_1;x_1,y_1)> y(x_1;\bar x, h(\bar x))$. Taking the supremum over $\bar{x}$ yields that $y_*(x_1)\leq y_1<+\infty$. In summary $y_*$ is defined pointwise on $[\sqrt{b_+/a},\infty)$, and $y_*<\infty$ is guaranteed.

Next we study the continuity and differentiability of $y_*$. Notice that by~\eqref{est: x<barx}, we know that $y_*(x) \geq h (x)\geq 0$ and is increasing on $[\sqrt{b/a}, \infty)$. Therefore, for every $x\in[\sqrt{b_+/a}, +\infty)$, we have 
$
y_*(x+) = \lim_{\epsilon\to 0^+} y_*(x+\epsilon) 
$
exists; and for every $x\in(\sqrt{b_+/a},+\infty)$, 
$
y_*(x-) = \lim_{\epsilon\to 0^+} y_*(x-\epsilon) 
$
exists as well. In particular, $f(x, y_*(x))$ is locally integrable on $[\sqrt{b_+/a},\infty)$. 

For $x_2>x_1\geq \sqrt{b_+/a}$, we want to estimate $y_*(x_2) - y_*(x_1)$. By the monotonicity of $y(x;\bar{x},h(\bar{x}))$ in $\bar{x}$ established in~\eqref{increasing in bar x}, we know that
\begin{align*}
y_*(x) = \sup\{y(x;\bar x,h(\bar x)): \bar x > \sqrt{b_+/a}\}.
\end{align*} 
Together with $y_*\geq 0$ and since $F$ is increasing on $[0,\infty)$, it follows that
\begin{align}\label{continuity upper}
y_*(x_2) - y_*(x_1) 
& =  \sup\{y(x_2;\bar x,h(\bar x)) - y_*(x_1): \bar x > \sqrt{b_+/a}\}\notag\\\notag
&\leq \sup\{y(x_2;\bar x,h(\bar x)) - y(x_1;\bar x,h(\bar x)): \bar x > \sqrt{b_+/a}\}\\\notag
& = \sup\left\{\int_{x_1}^{x_2} y'(\xi ;\bar x,h(\bar x)) \; d\xi : \bar x > \sqrt{b_+/a}\right\}\\\notag
&\leq \int_{x_1}^{x_2} \sup\{y'(\xi;\bar x,h(\bar x)):\bar x>\sqrt{b/a}\}\; d\xi \\\notag
& = \int_{x_1}^{x_2}  -a\xi^2 +b  + \sup\{F(y(\xi;\bar x,h(\bar x))):\bar x > \sqrt{b_+/a}\} \; d\xi \\\notag
& = \int_{x_1}^{x_2}  -a\xi^2 +b +F(y_*(\xi)) \; d\xi\\
&=  \int_{x_1}^{x_2} f(\xi, y_*(\xi)) d\xi.
\end{align}
For every $\delta>0$, there exists $\bar{x}$ such that $y(x_1;\bar x, h(\bar x)) +\delta > y_*(x_2)$. By~\eqref{increasing in bar x}, without loss of generality, we can assume that $\bar{x}>x_2$, and therefore $y(\xi ; \bar x, h(\bar x))$ is increasing in $\xi$ on the interval $[x_1, x_2]$ by~\eqref{est: x<barx}. Thus for every $\delta>0$ and for every $\xi\in[x_1, x_2]$, the monotonicity of $F$ on $[0,\infty)$ yields
$
F(y(\xi ; \bar x, h(\bar x))) \geq F(y(x_1 ; \bar x, h(\bar x))) \geq F(y_*(x_1) -\delta)
$.
Therefore,
\begin{align*}
y_*(x_2) - y_*(x_1) 
&\geq  y_*(x_2) -  (y(x_1;\bar x, h(\bar x))+\delta)\\
&\geq  y(x_2;\bar x, h(\bar x)) - y(x_1;\bar x, h(\bar x)) -\delta\\
& = \int_{x_1}^{x_2} y'(\xi;\bar x, h(\bar x))\; d\xi -\delta\\
&= \int_{x_1}^{x_2} -a\xi^2 +b+F(y(\xi;\bar x, h(\bar x))) d\xi -\delta\\
&\geq  (x_2-x_1) F(y(x_1 ; \bar x, h(\bar x)))-\delta + \int_{x_1}^{x_2} (-a\xi^2 +b)\; d\xi
\\&\geq  (x_2-x_1) F(y_*(x_1) - \delta) -\delta + \int_{x_1}^{x_2} (-a\xi^2 +b)\; d\xi.
\end{align*}
As this holds for arbitrary small $\delta>0$, it follows from the continuity of $F$ that
\begin{align}\label{continuity lower}
y_*(x_2) - y_*(x_1) \geq (x_2-x_1) F(y_*(x_1))+ \int_{x_1}^{x_2} (-a\xi^2 +b) d\xi.
\end{align}
By~\eqref{continuity upper}  and~\eqref{continuity lower}, we can conclude the continuity of $y_*$ on $[\sqrt{b_+/a},+\infty)$. Consider first that for $x\in[\sqrt{b_+/a},+\infty)$, by~\eqref{continuity upper} and the continuity of $y_*$,
\begin{align*}
\limsup_{\epsilon\to 0^+} \frac{y_*(x+\epsilon) - y_*(x) }{\epsilon} \leq \limsup_{\epsilon\to 0^+}  \frac{1}{\epsilon}  \int_{x}^{x+\epsilon} f(\xi, y_*(\xi)) d\xi = f(x,y_*(x)),
\end{align*}
and by~\eqref{continuity lower}.
\begin{align*}
\liminf_{\epsilon\to 0^+} \frac{y_*(x+\epsilon) - y_*(x) }{\epsilon} &\geq \liminf_{\epsilon\to 0^+} \frac{1}{\epsilon} \left( \epsilon F(y_*(x))+ \int_{x}^{x+\epsilon} (-a\xi^2 +b) d\xi \right)
\\&=  -ax^2 + b+F(y_*(x)) \\&=  f(x,y_*(x)). 
\end{align*}
In addition, for $x\in(\sqrt{b_+/a},+\infty)$, by~\eqref{continuity upper},~\eqref{continuity lower} and the continuity of $y_*$,
\begin{align*}
f(x,y_*(x))
&= \limsup_{\epsilon\to 0^+}  \frac{1}{\epsilon}  \int_{x-\epsilon}^x f(\xi, y_*(\xi)) d\xi \\
&\geq \limsup_{\epsilon\to 0^+} \frac{y_*(x) - y_*(x-\epsilon) }{\epsilon}\\
&\geq \liminf_{\epsilon\to 0^+} \frac{y_*(x) - y_*(x-\epsilon) }{\epsilon}\\
&\geq \liminf_{\epsilon\to 0^+} \frac{1}{\epsilon} \left( \epsilon F(y_*(x-\epsilon))+ \int_{x-\epsilon}^x (-a\xi^2 +b) d\xi \right)\\
&= -ax^2 + b+F(y_*(x)) \\&=  f(x,y_*(x)).
\end{align*}
Hence we can conclude that $y_*'$ exists and 
\begin{align*}
y_* '(x) = f(x,y_*(x)), \mbox{ for all } x \in [\sqrt{b_+/a},+\infty).
\end{align*}
In summary, for $b\in\R$, the function $y_*$ therefore is a solution of~\eqref{ode:general} that satisfies properties~\ref{positive},~\ref{domain} and hence satisfies also the growth condition~\eqref{lemma_growth}. 

We only need to show~\eqref{eq:b} holds. When $b>0$, where $b_+=b$, $[0,\sqrt{b/a})$ is contained in the maximal interval of existence of $y_*$ is a side product in the proof of~\eqref{eq:b}

We rewrite $h(x;b) = F^{-1}(ax^2-b)$, and notice that $h'$ is strictly decreasing by Lemma~\ref{prop2} (iv). Let $y(x;\bar x, h(\bar x ;b))$ denote the solution to~\eqref{ode:general} with constant $b$ and initial condition $(\bar x, h(\bar x;b))$ with $\bar x >\sqrt{b_+/a}$. Then by Proposition 2.76 in~\cite{chicone1999}, we know that
\begin{align}\label{ivp continuity wrt data}
\frac{\partial}{\partial b} y(x;\bar x, h(\bar x ;b)) = \Phi(x; \bar x, b), 
\end{align}
where $\Phi(x; \bar x, b)$ is the solution to the following initial value problem on its maximal interval of existence, 
\begin{equation}\label{ivp wrt data}
\left\{
\begin{aligned}
\Phi'(x; \bar x, b) &= F'\left(y(x;\bar x, h(\bar x ;b)) \right) \Phi(x; \bar x, b), \\
\Phi(\bar x; \bar x, b) &= \frac{\partial}{\partial b} h(\bar x; b) = -\left(F^{-1}\right)' (a\bar x^2 - b). 
\end{aligned}
\right.
\end{equation}
Indeed, for $\bar x > x\geq \sqrt{b_+/a}$, $y(x;\bar x, h(\bar x ;b))\geq h(x ;b) \geq 0$, variation of constants yields:
\begin{align}\label{b not positive}
 \Phi(x; \bar x, b) =  -\left(F^{-1}\right)' (a\bar x^2 - b) e^{-\int_x^{\bar x} F'\left(y(\xi;\bar x, h(\bar x ;b)) \right) \; d\xi} \leq 0.
\end{align}
By~\eqref{increasing in bar x} and Lemma~\ref{prop2} (iv), we conclude that $ \Phi(x; \bar x, b)$ is increasing in $\bar x$ on $(x, +\infty)$. 
We rewrite $y(x;b)$ as the unique solution to~\eqref{ode:general} with constant $b$ where $[\sqrt{b_+/a},+\infty)$ is contained in its maximal interval of existence $D_b$, such that $y(x;b)\geq 0$ for every $x\geq \sqrt{b_+/a}$. For $x\geq \sqrt{b_+/a}$, we claim $\frac{\partial}{\partial b} y(x;b)$ exists and 
\begin{align}\label{continuity wrt b}
\frac{\partial}{\partial b} y(x;b) = \sup\left\{\Phi(x; \bar x, b): \bar x > x\right\} \leq 0. 
\end{align}
To this end, for $b_2>b_1$, fix $x\geq \sqrt{{b_2}_+/a}$, first observe that
\begin{align}\label{b upper}
y(x;b_2) - y(x;b_1) 
&= \sup\left\{ y(x;\bar x, h(\bar x ;b_2)) - y(x;b_1) : \bar x > x \right\}\notag\\
&\leq \sup\left\{ y(x;\bar x, h(\bar x ;b_2)) - y(x;\bar x, h(\bar x ;b_1)) : \bar x > x \right\}\notag\\
&= \sup\left\{ \int_{b_1}^{b_2} \frac{\partial}{\partial b} y(x;\bar x, h(\bar x ;b)) \; db : \bar x > x \right\}\notag\\
&\leq  \int_{b_1}^{b_2} \sup\left\{\frac{\partial}{\partial b} y(x;\bar x, h(\bar x ;b))  : \bar x > x \right\}db\notag\\
&=  \int_{b_1}^{b_2} \sup\left\{ \Phi(x; \bar x, b)  : \bar x > x \right\}db.
\end{align}
On the other hand, for every $\delta>0$, there exists $\bar x_\delta $ such that 
\begin{align*}
y(x;\bar x_\delta, h(\bar x_\delta ;b_1)) +\delta >y(x;b_1). 
\end{align*}
By~\eqref{increasing in bar x}, we can assume without loss of generality that $\bar x_\delta > x$ and, for every $\bar x >\bar x_\delta $,
\begin{align*}
y(x;\bar x, h(\bar x ;b_1)) +\delta >y(x;b_1).
\end{align*}
Therefore, for every $\bar x >\bar x_\delta$,
\begin{align*}
y(x;b_2) - y(x;b_1) &\geq y(x;b_2) - y(x;\bar x, h(\bar x ;b_1))-\delta\\
& \geq y(x;\bar x, h(\bar x ;b_2)) -y(x;\bar x, h(\bar x ;b_1)) -\delta\\
& = \int_{b_1}^{b_2} \frac{\partial}{\partial b} y(x;\bar x, h(\bar x ;b))db - \delta\\
& =  \int_{b_1}^{b_2} \Phi(x; \bar x, b)db -\delta.
\end{align*}
As $\Phi(x; \bar x, b) $ is increasing in $\bar x$ on $(x, +\infty)$, it follows that
\begin{align*}
y(x;b_2) - y(x;b_1) &\geq \sup\left\{ \int_{b_1}^{b_2}  \Phi(x; \bar x, b)  \; db : \bar x > \bar x_\delta \right\}-\delta \\
& =  \int_{b_1}^{b_2}  \sup\left\{ \Phi(x; \bar x, b)  : \bar x > \bar x_\delta \right\}db -\delta\\
&=\int_{b_1}^{b_2}  \sup\left\{ \Phi(x; \bar x, b)  : \bar x > x \right\}db -\delta.
\end{align*}
Since the above inequality holds for every $\delta>0$, we conclude that
\begin{align}\label{b lower}
y(x;b_2) - y(x;b_1) \geq \int_{b_1}^{b_2}  \sup\left\{ \Phi(x; \bar x, b)  : \bar x > x \right\}db. 
\end{align}
By~\eqref{b upper} and~\eqref{b lower}, $y(x;b)$ is continuous and differentiable with respect to $b$. We can infer our claim~\eqref{continuity wrt b} from~\eqref{b not positive}; in particular, \eqref{eq:b} holds. By Theorem 2.77 in~\cite{chicone1999}, $[0,\infty)$ is contained in the maximal interval of existence of $y(\cdot;b)$ for every $b\in\R$. 

In summary, for every $a>0$, $b\in\R$, we have shown that there exists a unique solution $y_*(\cdot;b)$ of~\eqref{ode:general} with $y_*(x;b)\geq 0$ on $[\sqrt{b_+/a},\infty)$. Moreover, we have derived a number of properties of $y_*$ that will be utilized below.
\end{proof}

In Lemma~\ref{general ODE}, we have shown that for every $b\geq0$, there exists non-negative solution $y_r$ to~\eqref{ode:general} on $[0,\infty)$. A symmetric argument yields that for every $b\geq0$, there exists a non-positive solution $y_l$ to~\eqref{ode:general} on $(-\infty, 0]$. Then by the monotonicity of $y(0;b)$ with respect to $b$, there exists a unique choice of the constant $b$ in \eqref{ode:general} that allows to smoothly paste together the solution $y_l$ and $y_r$ at 0, thereby obtaining a solution of~\eqref{ode:general} on the whole real line. 

\begin{lemma}\label{symmetric}
Let $F$ be a function satisfying Assumption~\ref{property}. Then there exists a unique constant $b_F>0$ such that when $b=b_F$, the ODE
\begin{align}
g'(x) = ax^2 - b - F(g(x)), \tag{\ref{ODE:g}}
\end{align}
has a solution $g$ on $\R$ such that $xg(x)\leq 0$. 
Moreover, $g$ is unique, and it is odd and decreasing and satisfies the following growth conditions: 
\begin{align}\label{moments}
\lim_{x\to -\infty} \frac{g(x)}{F^{-1}(ax^2)} = 1, \quad \quad
\lim_{x\to +\infty} \frac{g(x)}{F^{-1}(ax^2)} = -1. 
\end{align}
Further, there exists $K>0$, such that for $x\in\R$,
\begin{align*}
|g(x)| \leq K(1+|x|), \qquad |g'(x)| \leq K.
\end{align*}
\end{lemma}

\begin{proof}
From Lemma~\ref{general ODE}, we know that for every parameter $b\geq0$ there exists a unique solution $y_r(x;b)$ with $y_r(x;b)\geq 0$ for every $x\geq \sqrt{b/a}$. Moreover, $[0,\infty)$ contains in the maximal existence of interval of $y_r(\cdot;b)$, i.e. $0\in D_b$, and $y_r(x;b)$ satisfies
$$
\lim_{x\to+\infty} \frac{y_r(x;b)}{F^{-1}(ax^2)} = 1.
$$
 Define $y_l(x;b) = - y_r(-x;b)$ on $(-\infty,0]$. Then
$$
\lim_{x\to -\infty} \frac{y_l(x;b)}{F^{-1}(ax^2)} = -\lim_{x\to- \infty} \frac{y_r(-x;b)}{F^{-1}(ax^2)}= -\lim_{x\to \infty} \frac{y_r(x;b)}{F^{-1}(ax^2)} = -1.
$$
Moreover, since $F$ is even, for $x\leq 0$,
\begin{align}\label{left}
{y_l}'(x;b) 
&= {y_r}'(-x;b) =- a(-x)^2 + b + F(y_r(-x;b)) \notag\\
&= -ax^2 + b + F(-y_r(-x;b)) \nonumber
\\&= -ax^2 +  b + F(y_l(x;b)).
\end{align}
That is, $y_l(x;b)$ also satisfies~\eqref{ode:general} on $(-\infty,0]$. 

For $b=0$, by Item~\ref{positive} in the proof of Lemma~\ref{general ODE}, we have $y_r(x;0)> F^{-1}(ax^2)$. Hence, by the continuity of $y_r$ and $F^{-1}$,
\begin{align*}
y_r(0;0) \geq F^{-1}(0) =0 \geq -y_r(0;0) = y_l(0;0).
\end{align*}
By~\eqref{eq:b} in Lemma~\ref{general ODE}, for $x\geq0$, $y_r(x;b)$ is strictly decreasing in $b$ and thus $y_r(x;b)\leq y_r(x;0)<\infty$ for all $b\geq 0$. In addition, we claim that as $b\to+\infty$, $y_r(0;b)$ goes to $ -\infty$. Suppose not. Then there exists 
\begin{align*}
\delta_1=\lim_{b\to+\infty} y_r(0;b) >-\infty. 
\end{align*}
As a result, 
\begin{align*}
y_r(1;b) &= y_r(0;b) + \int_0^1 -ax^2 + b + F(y(x;b))~dx \geq   y_r(0;b) + \int_0^1 (-a + b) dx \geq \delta_1 + b - a,
\end{align*}
and, for $b\to+\infty$,
\begin{align*}
y_r(1;0) \geq \lim_{b\to+\infty} y_r(1;b) \geq  \lim_{b\to+\infty} \delta_1 + b - a = +\infty,
\end{align*}
which leads to contradiction. Hence as $b\to+\infty$, $y_r(0;b)$ goes to $ -\infty$, and $y_l(0;b) = -y_r(0;b)$ goes to $+\infty$.
Thus, for some constant $b_F\geq 0$ we have 0 is contained in $D_{b_F}$ and 
\begin{align}\label{zero}
y_r(0;b_F) =0= y_l(0;b_F).
\end{align}
As $y_r(x;b)$ is decreasing in $b$, the constant $b_F$ is unique. 

Now we use $y_r(\cdot;b_F)$ and $y_l(\cdot;b_F)$ to construct the solution for~\eqref{ODE:g}:
\begin{align*}
g(x) = -y_r(x;b_F) \one_{\{x\geq0\}}-y_l(x;b_F) \one_{\{x<0\}}.
\end{align*}
It's easy to see that $g$ is defined on $\mathbb{R}$ and satisfies the growth conditions~\eqref{moments}. We now show that $g$ is indeed a solution of~\eqref{ODE:g} with $b=b_F$. 
Using~\eqref{zero}, we can see that $g$ is continuous and equal to zero at $x=0$. Therefore,
\begin{align*}
g(x) = -y_r(x;b_F) \one_{\{x\geq0\}}+y_r(-x;b_F) \one_{\{x<0\}} = -y_r(x;b_F) \one_{\{x>0\}}+y_r(-x;b_F) \one_{\{x\leq 0\}}
= -g(-x),
\end{align*}
which implies that $g$ is odd. Furthermore, as $y_r$ is increasing on $[\sqrt{b_F/a},\infty)$, and 
\begin{align*}
y_r'(x;b_F) = -ax^2+b_F + F(y_r(x;b_F)) \geq -ax^2+b_F \geq 0, \quad \mbox{for $x\in[0,\sqrt{b_F/a}]$,}
\end{align*}
$y_r$ is increasing on $[0,\infty)$, and we infer that $g$ is decreasing.  Since $F$ is even, we have 
\begin{align*}
F(g(x)) = F(-y_r(x;b_F)) = F( -y_l(-x;b_F)) = F(g(-x)), \quad \mbox{for $x\geq 0$.}
\end{align*}
Therefore we can conclude that
\begin{align*}
g'(x) = -{y_r}'(x) = ax^2 - b_F - F(y_r(x;b_F)) = ax^2 - b_F - F(g(x)), \quad \mbox{for $x> 0$.}
\end{align*}
Likewise,
\begin{align*}
g'(x) = -{y_l}'(x) = ax^2 - b_F - F(y_l(x;b_F)) = ax^2 - b_F - F(g(x)), \quad \mbox{for $x<0$.}
\end{align*}
Moreover, the continuity of $g'$ is guaranteed at $x=0$ since 
\begin{align*}
\lim_{x\to 0^+} {g}'(x) = -{y_r}'(0;b_F) = -b_F  = -{y_l}'(0;b_F)= \lim_{x\to 0^-} g'(x).
\end{align*}
In summary, $g$ therefore is indeed a solution of~\eqref{ODE:g} with $b=b_F$. 

Next, we show that $g$ is unique. 
Let $g$, defined and continuously differentiable on $\R$, satisfy~\eqref{ODE:g} for some $b\geq 0$ and also satisfy $xg(x)\leq 0$ for $x\in\R$. Then $-g$ is the unique function $y(\cdot;b)$ in Lemma~\ref{general ODE}. Because $F$ is even and $g$ satisfies~\eqref{ODE:g}, we know $g(-x)$ also satisfies the conditions of Lemma~\ref{general ODE}. Hence $g(-x) = y(x;b)$ for $x$ in the maximal interval of existence of $y(\cdot;b)$. Therefore, 
\begin{align*}
-g(0) = y(0;b) = g(0),
\end{align*}
which implies $y(0;b) = 0$. This forces $b$ to be equal to $b_F$, and $g$ to be the function constructed above. 

The growth condition~\eqref{moments} and Lemma~\ref{prop2} (v) imply that there exist $x_0>0$ and $\hat{c}>0$ such that, for every $|x|>x_0$, 
\begin{align*}
|g(x)| 
\leq 2|F^{-1}(ax^2)|\leq 2\hat{c}|x|.
\end{align*}
Therefore, for all $x$, and since $-g$ is increasing,  
\begin{align}\label{sublinear of g}
|g(x)| \leq  |g(x_0)|+2\hat{c}|x| .
\end{align}

Now we establish the boundedness of $g'$, using an idea from~\cite{bayraktar.al.18}. Since $g$ is odd, we only need to show that for $x>0$, $g'$ is bounded from below. From~\eqref{sublinear of g}, we can infer that as $x\to\infty$, $g'$ cannot go to $-\infty$. Therefore, there exists $M>0$ and an increasing sequence $\{x_n\}_{n=1}^\infty$ such that $x_n\to\infty$ and 
$-M\leq g'(x_n)\leq 0$. Now suppose $g'$ is not bounded from below, which means that for every integer $n>M$, there exists $z_n > x_n$ such that $g'(z_n)\leq-n$. For each $n>M$, let $m(n)>n$ denote the first integer such that $x_n<z_n<x_{m(n)}$. Then from
\begin{align*}
g'(z_n)\leq -n<-M \leq \min\{g'(x_n),  g'(x_{m(n)})\},
\end{align*}
we can infer that there exists a local minimum of $g'$ on $[x_n, x_{m(n)}]$ for every integer $n>M$, denoted by $\xi_n$. Therefore, for every integer $n>M$, $g''(\xi_n) = 0$, and
\begin{align*}
0 \leq g'''(\xi_n) &= 2a - F''(g(\xi_n))\big(g'(\xi_n)\big)^2 - F'(g(\xi_n))g''(\xi_n) = 2a - F''(g(\xi_n))\big(g'(\xi_n)\big)^2.
\end{align*}
Together with Assumption~\ref{property} (iv), we know that $F''(g(\xi_n))>0$ for $n$ large enough, and hence 
\begin{align*}
n^2 \leq \big(g'(\xi_n)\big)^2  \leq \frac{2a}{F''(g(\xi_n))} \leq \frac{2a}{\tilde{C}},
\end{align*}
which leads to a contradiction. Without loss of generality, we choose $M>0$ large enough so that $|g'(x)|<M$ for every $|x|>x_0$. Now, choose $K>M+|g(x_0)|+2\hat{c}$. Then we have
\begin{align*}
|g(x)| \leq K(1+|x|), \qquad |g'(x)| \leq K
\end{align*}
as asserted. This completes the proof.
\end{proof}

Next, we show that with $b=b_F$, the solution to the first-order ODE~\eqref{ODE:g} on $\R$ with $xg(x)\leq0$ is also the unique solution on $\R$ to the second-order ODE~\eqref{2nd ODE:g} with $xg(x)\leq0$. 

\begin{lemma}\label{2nd ODE}
Let $F$ be a function satisfying Assumption~\eqref{property}. Then the unique solution $g$ on $\R$ to~\eqref{ODE:g} such that $xg(x)\leq 0$ is also the unique solution on $\mathbb{R}$ of the second-order ODE
\begin{equation}\tag{\ref{2nd ODE:g}}
g''(x) = 2ax - F'(g(x))g'(x)
\end{equation}
such that $xg(x)\leq 0$.
\end{lemma}

\begin{proof}
In view of the first-order ODE~\eqref{ODE:g} satisfied by $g$, its derivative is also differentiable. Differentiating the ODE for $g$ in turn shows that $g$ also satisfies the second-order ODE~\eqref{2nd ODE:g}.

Now suppose $\tilde{g}$ is a solution of the second-order ODE~\eqref{2nd ODE:g} satisfying $x\tilde{g}(x)\leq 0$, hence we can infer that $\tilde{g}$ is non-increasing at zero. 
As 
\begin{align*}
\left(F(\tilde{g}(x))\right)' = F'(\tilde{g}(x)) \tilde{g}'(x), 
\end{align*}
integrating both sides of~\eqref{2nd ODE:g} gives 
\begin{align*}
\tilde{g}'(x) = \tilde{g}'(0) +  \int_0^x \Big(2a\xi - F'(\tilde{g}(\xi))\tilde{g}'(\xi)\Big) d\xi = ax^2 - \tilde{b} - F(\tilde{g}(x)),
\end{align*}
for some constant $\tilde{b} = F(\tilde{g}(0))-\tilde{g}'(0)$. By Lemma~\ref{symmetric} we know $b_F\geq 0$ is the unique constant such that~\eqref{ODE:g} has a solution on $\R$ with $xg(x)\leq 0$. Thus, $\tilde{b}=b_F$ and, by the uniqueness of $g$, we have $\tilde{g} = g$. This completes the proof. 
\end{proof}

We introduce one more lemma before the proof of Lemma~\ref{ODE}. 
\begin{lemma}\label{equ assump}
Suppose the general cost function $G$ satisfies Assumption~\ref{cond:cost}. Then $G^*$, the Legendre transform of $G$, satisfies Assumption~\ref{property}, and so does $cG^*(\frac{x}{c})$, where $c>0$ is a constant. 
\end{lemma}

\begin{proof}
Observe that the Legendre transformation of the cost function $G(x)$ is
\begin{align*}
G^*(x) =  x(G')^{-1}(x) - G((G')^{-1}(x)).
\end{align*}
Since the instantaneous cost $G$ is even, $G'$ and in turn $(G')^{-1}$ are odd, so that the function $G^*$ is even. Moreover, $G(0)=G'(0)=0$ imply $G^*(0)=0$. As both $G$ and $(G')^{-1}$ are differentiable, 
\begin{align*}
(G^*)'(x)=(G')^{-1}(x)>0.
\end{align*} 
In particular, $(G^*)^{-1}$ exists on $[0,\infty)$ and is differentiable. Moreover, by the convexity and twice differentiability of $G$,
\begin{align*}
(G^*)''(x) =((G')^{-1})'(x) >0.
\end{align*}
It follows that $G^*$ is convex and $(G^*)'$ is strictly increasing, so that Assumptions~\ref{property} (i,ii) are satisfied.  
By Assumption~\ref{cond:cost}, $|(G')^{-1}(x)|\leq C(1+|x|^{k-1})$ for $C>0$ and $k\geq 2$. Whence, there exists a constant $K>0$ such that
\begin{align*}
G^*(x)=|G^*(x)| \leq |x(G')^{-1}(x)|\leq C(|x| + |x|^k)\leq K(1+|x|^k).
\end{align*}
Therefore, Assumption~\ref{property}(iii) is also satisfied. Again by Assumption~\ref{cond:cost}, $\left(G'\right)^{-1}$ is increasing, and there exists $C>0$ and $x_0>0$, such that for large $x>x_0$, $\left(G'\right)^{-1}(x)$ is large and by Assumption~\ref{cond:cost}(iii)
\begin{align*}
(G^*)'' (x) = ((G')^{-1})'(x) = \frac{1}{G''\left( \left(G'\right)^{-1}(x)\right)}\geq \frac{1}{C}.
\end{align*}
Thus, Assumption~\ref{property}(iv) holds as well. 
\end{proof}

We now turn to the proof of Lemma~\ref{ODE}. 

\begin{proof}[Proof of Lemma~\ref{ODE}]
Let $G^*$ denote the Legendre transform of $G$, and define
\begin{align*}
a = \frac{\gamma \sigma^2}{\delta^2}, \quad F(x) = \frac{2}{\delta^2}G^*(x), 
\end{align*}
where $\gamma$ and $\delta$ are defined as in Lemma~\ref{ODE}. By Lemma~\ref{equ assump}, $G^*$ and in turn $F$ satisfy Assumption~\ref{property}. For the above choices of $a$ and $F$,  Lemma~\ref{general ODE} and Lemma~\ref{symmetric} therefore yield the existence and uniqueness of the constant $b_F$ and the solution $g$ on $\R$ to the first-order ODE~\eqref{ODE:g} such that $xg(x)\leq 0$ for every $x\in\R$. In view of the first-order ODE~\eqref{ODE:g} satisfied by $g$, 
\begin{align*}
g'(x) &
=\frac{\gamma\sigma^2}{\delta^2} x^2 -F(g(x))- b_F
=\frac{\gamma\sigma^2}{\delta^2} x^2-\frac{2}{\delta^2}\big[ g(x)(G')^{-1}(g(x)) - G((G')^{-1}(g(x)))\big] - b_F,
\end{align*}
Lemma~\ref{2nd ODE} shows that $g$ is also the unique solution to the ODE~\eqref{eqn:ergodic ODE} from Lemma~\ref{ODE}:
\begin{align*}
\frac 1 2 \delta^2 g''(x) = -\big[ g(x)(G')^{-1}(g(x)) - G((G')^{-1}(g(x)))\big]' + \gamma\sigma^2 x = -g'(G')^{-1}(g(x)) + \gamma\sigma^2 x.
\end{align*}
To complete the proof, notice that 
\begin{align*}
F^{-1}(ax^2) = F^{-1}\left(\frac{\gamma\sigma^2}{\delta^2}x^2\right) = (G^*)^{-1}\left(\frac{\delta^2}{2}\frac{\gamma\sigma^2}{\delta^2}x^2\right) = (G^*)^{-1} \left(\frac{\gamma\sigma^2}{2}x^2\right),
\end{align*}
which yields the analogue of the growth conditions~\ref{moments}:
\begin{equation}\label{conditions}
\lim_{x\to -\infty} \frac{g(x)}{(G^*)^{-1}(\frac{\gamma\sigma^2}{2}x^2)} = 1, \quad \quad
\lim_{x\to +\infty} \frac{g(x)}{(G^*)^{-1}(\frac{\gamma\sigma^2}{2}x^2)} = -1. 
\end{equation}
\end{proof}

\section{Calibration Details}\label{app:calculation}

In this section, we provide some additional details concerning the calibration of the model with costs of general power form at the end of Section~\ref{ss:calcosts}. If $G_q(x)=\lambda_q|x|^q/q$ with $\lambda_q>0$, $q \in (1,2]$, then the nonlinear ODE~\eqref{eqn:ergodic ODE} from Lemma~\ref{ODE} can be simplified by rescaling. Indeed, the solution then can be written as 
\begin{align}\label{rescale solution}
g_q(x) = \left(\frac{\lambda_q}{q}\right)^{\frac{3}{q+2}}\left(\frac{\gamma\sigma^2\delta_q^4}{8}\right)^{\frac{q-1}{q+2}} \tilde{g}_q\left(  2^{\frac{q-1}{q+2}} \left(\frac{q\gamma\sigma^2}{\lambda_q}\right)^{\frac{1}{q+2}} \delta_q^{-\frac{2q}{q+2}} x\right ), 
\end{align}
where $\tilde{g}_q$ is the unique solution on $\R$ of\footnote{As shown in Lemma~\ref{2nd ODE}, $\tilde{g}_q$ is in fact the solution to the first-order equation~(17) in \cite[Theorem~6]{guasoni.weber.18}, with $q = \alpha+1$.}
\begin{align}\label{canonical}
\tilde{g}_q''(x) + \tilde{g}_q'(x) \sign(\tilde{g}_q(x)) \left|q^{-1}\tilde{g}_q(x)\right|^{\frac{1}{q-1}}= 2x. 
\end{align}
This rescaled ODE only depends on the elasticity $q$ of the trading cost but not on the other primitives of the model. As a consequence, the rescaled ODE only needs to be solved numerically once for each $q$ to determine the transaction cost $\lambda_q$ and $\delta_q$ that match the variance of the state variable for proportional costs and the average share turnover observed empirically. To this end, first notice that 
\begin{align}\label{drift for X}
&\left(G_q'\right)^{-1} \left({g_q(x)}\right)= -\sign(x)\left|\frac{g_q(x)}{\lambda_q}\right|^{\frac{1}{q-1}} \notag\\
&\quad = -\left(\frac{q\gamma\sigma^2\delta_q^4}{8\lambda_q}\right)^{\frac{1}{q+2}}\sign(x)\left| q^{-1}\tilde{g}_q\left(  2^{\frac{q-1}{q+2}} \left(\frac{q\gamma\sigma^2}{\lambda_q}\right)^{\frac{1}{q+2}} \delta_q^{-\frac{2q}{q+2}} x\right )\right|^{\frac{1}{q-1}}.
\end{align}
For power costs $G_q(x)=\lambda_q |x|^q/q$,  the dynamics of the state-variable $X$ from  Lemma~\ref{lem:sde} in turn are given by
\begin{align} \label{dynamic for X_q}
dX_t &= -\left(\frac{q\gamma\sigma^2\delta_q^4}{8\lambda_q}\right)^{\frac{1}{q+2}}\sign(X_t)\left|q^{-1}\tilde{g}_q\left(  2^{\frac{q-1}{q+2}} \left(\frac{q\gamma\sigma^2}{\lambda_q}\right)^{\frac{1}{q+2}} \delta_q^{-\frac{2q}{q+2}} X_t \right)\right|^{\frac{1}{q-1}} dt+\delta_q dW_t.
\end{align}
The stationary density of the state variable $X$ therefore can therefore be computed via the normalized speed measure as\footnote{For quadratic costs $q=2$, where the state variable has Ornstein-Uhlenbeck dynamics, this reduces to the density of the normal distribution.} 
\begin{align*}
{\nu}_q(x) =  \frac{2^{\frac{q-1}{q+2}}(\frac{q\gamma\sigma^2}{\lambda_q})^{\frac{1}{q+2}}\delta_q^{-\frac{2q}{q+2}}
\exp\Big(-\int_0^{2^{\frac{q-1}{q+2}}(\frac{q\gamma\sigma^2}{\lambda_q})^{\frac{1}{q+2}}\delta_q^{-\frac{2q}{q+2}} x} |q^{-1}\tilde{g}_q(y)|^{\frac{1}{q-1}}dy\Big)}{ 2\int_0^\infty \exp\Big(- \int_0^{x} |q^{-1}\tilde{g}_q(y)|^{\frac{1}{q-1}}dy\Big) dx }.
\end{align*}

The goal now is to choose the model parameters $\lambda_q$ and $\delta_q$ to match the share turnover in the model to its empirical level and the stationary variance of the state variable to its counterpart for proportional costs. To this end, define 
\begin{align*}
 \tilde{c}_q &=  \left[  2\int_0^\infty \exp\left(- \int_0^{x} \left|q^{-1}\tilde{g}_q(y)\right|^{\frac{1}{q-1}}dy\right) dx \right]^{-1}, \\
   \tilde{v}_q &=  2  \tilde{c}_q \int_0^\infty x^2 \exp\left(- \int_0^{x} \left|q^{-1}\tilde{g}_q(y)\right|^{\frac{1}{q-1}}dy\right) dx.
\end{align*}
To match the total share turnover, we then need
\begin{align}\label{matching turnover}
\text{ShTu} 
=\int_{-\infty}^ \infty \left|\frac{g_q(x)}{\lambda_q}\right|^{\frac{1}{q-1}}\nu_q(x) dx 
= 2^{\frac{q-1}{q+2}} \left(\frac{q\gamma\sigma^2}{\lambda_q}\right)^{\frac{1}{q+2}} \delta_q^{-\frac{2q}{q+2}}\tilde{c}_q\delta_q^2,
\end{align}
which is satisfied if we choose
\begin{equation}\label{eq:deltaq}
\delta_q=\left(\frac{\lambda_q}{2^{q-1}q\gamma\sigma^2}\left(\frac{\text{ShTu}}{\tilde{c}_q}\right)^{q+2}\right)^{1/4}.
\end{equation}
After matching the average share turnover, we now choose the size $\lambda_q$ of the trading cost to match the stationary variance of the state variable to its counterpart for proportional costs $\lambda_1$. For power costs with elasticity $q$, the stationary mean of the state variable is zero by the symmetry of $\nu_q$, and the stationary variance can in turn be computed by integrating against the stationary density, 
\begin{align}\label{eq:varianceq}
2 \int_0^\infty x^2 \nu_q(x) dx 
=\frac{\tilde{v}_q}{4^{\frac{q-1}{q+2}} \left(\frac{q\gamma\sigma^2}{\lambda_q}\right)^{\frac{2}{q+2}} \delta_q^{-\frac{4q}{q+2}}}= \frac{\tilde{v}_q \lambda_q}{2^{q-1}q\gamma\sigma^2} \left(\frac{\text{ShTu}}{\tilde{c}_q}\right)^q,
\end{align}
where we have inserted~\eqref{eq:deltaq} in the second step. For proportional costs the state variable is a doubly reflected Brownian motion, whose stationary law is the uniform distribution on $[-l,l]$ (which has variance $l^2/3$). Recall from Section~\ref{ss:calcosts} that for proportional costs $\lambda_1$, we need 
\begin{align*}
\delta_1 = \frac{\xi_1}{\sigma} = \left(\frac{12\mathrm{ShTu}^3\lambda_1}{\gamma\sigma^2}\right)^{1/4},
\end{align*}
to match the average share turnover $\text{ShTu}$.  After inserting this into Formula~\eqref{eq:ell} for the trading boundary $l$, it follows that the stationary variance $l^2/3$ of the state variable for proportional costs is given by
\begin{align*}
\frac{l^2}{3}= \frac13 \left(\frac{3\lambda_1 \delta_1^2}{2\gamma\sigma^2}\right)^{2/3}=\frac{\lambda_1\text{ShTu}}{\gamma\sigma^2}.
\end{align*}
To match this with the corresponding stationary variance~\eqref{eq:varianceq} for power costs with elasticity $q$, we therefore need to choose
\begin{align}\label{est:lam}
\lambda_q = \frac{q \tilde{c}_q}{\tilde{v}_q}\left(\frac{2 \tilde{c}_q}{\text{ShTu}}\right)^{q-1}\lambda_1.
\end{align}
In summary, for a given value of $q$, the solution $\tilde{g}_q$ of~\eqref{canonical} therefore needs to be computed numerically on a fine grid once. Then, we can use numerical integration to determine $\tilde{c}_q$, $\tilde{v}_q$. This finally allows to compute the $\lambda_q$ corresponding to the proportional trading cost $\lambda_1$ via~\eqref{est:lam}, and pins down the corresponding $\delta_q$ through~\eqref{eq:deltaq}.

\bibliographystyle{abbrv}
\bibliography{references}

\end{document}